\documentclass[reqno]{amsart}
\usepackage{hyperref}
\usepackage{graphicx}
\usepackage{curves}
\usepackage{float}
\usepackage{caption}
\usepackage{subcaption}
\usepackage{dirtytalk}
\usepackage{amsthm}
\usepackage{amsbsy}
\usepackage{mathtools}
\usepackage{cool}
\theoremstyle{plain}
\usepackage{tikz-cd}
\usetikzlibrary{arrows}
\usetikzlibrary{decorations.text}
\usepackage{pgfplots}
\pgfplotsset{compat=1.14}

\newcommand*{\mailto}[1]{\href{mailto:#1}{\nolinkurl{#1}}}

\DeclareMathOperator{\arcosh}{arcosh}
\newcommand\restr[2]{{
  \left.\kern-\nulldelimiterspace 
  #1 
  \vphantom{\big|} 
  \right|_{#2} 
  }}
\DeclarePairedDelimiter{\floor}{\lfloor}{\rfloor}

\newtheorem{theorem}{Theorem}[section]
\newtheorem{definition}[theorem]{Definition}
\newtheorem{assumption}[theorem]{Assumption}
\newtheorem{lemma}[theorem]{Lemma}
\newtheorem{corollary}[theorem]{Corollary}
\newtheorem{proposition}[theorem]{Proposition}
\newtheorem{remark}[theorem]{Remark}

\newcommand{\R}{{\mathbb R}}

\newcommand{\Z}{{\mathbb Z}}
\newcommand{\C}{{\mathbb C}}

\newcommand{\nn}{\nonumber}
\newcommand{\be}{\begin{equation}}
\newcommand{\ee}{\end{equation}}
\newcommand{\bea}{\begin{eqnarray}}
\newcommand{\eea}{\end{eqnarray}}
\newcommand{\btheo}{\begin{theorem}}
\newcommand{\etheo}{\end{theorem}}

\newcommand{\ol}{\overline}
\newcommand{\ti}{\tilde}
\newcommand{\hb}{\hbar}

\newcommand{\im}{\Im}
\newcommand{\asympt}{\mathcal{O}}
\newcommand{\var}{\mathcal{V}}
\newcommand{\Ai}{\operatorname{Ai}}
\newcommand{\Bi}{\operatorname{Bi}}
\newcommand{\W}{\mathcal{W}}

\newcommand{\esf}{\mathsf{E}}
\newcommand{\msf}{\mathsf{M}}
\newcommand{\nsf}{\mathsf{N}}
\newcommand{\xsf}{\mathsf{X}}
\newcommand{\rsf}{\boldsymbol\rho}
\newcommand{\musf}{\boldsymbol\mu}

\newcommand{\eps}{\varepsilon}
\newcommand{\vphi}{\varphi}
\newcommand{\s}{\sigma}
\newcommand{\lam}{\lambda}
\newcommand{\g}{\gamma}
\newcommand{\G}{\Gamma}
\newcommand{\om}{\omega}
\newcommand{\Om}{\Omega}
\newcommand{\z}{\zeta}
\newcommand{\m}{\mu}
\newcommand{\al}{\alpha}
\newcommand{\ta}{\tau}
\newcommand{\ps}{\psi}
\newcommand{\8}{\theta}
\newcommand{\thv}{\vartheta}
\newcommand{\ro}{\rho}
\newcommand{\fisf}{\mathsf{\Phi}}
\newcommand{\de}{\delta}
\newcommand{\La}{\Lambda}

\numberwithin{equation}{section}

\begin{document}

\title[SEMICLASSICAL WKB PROBLEM FOR THE
NON-SELF-ADJOINT DIRAC OPERATOR]
{SEMICLASSICAL WKB PROBLEM FOR THE
NON-SELF-ADJOINT DIRAC OPERATOR WITH 
A DECAYING POTENTIAL }

\author[N. Hatzizisis]{Nicholas Hatzizisis $^{\dag}$}
\address{$^{\dag}$ Mathematics Building, University of Crete \\
700 13 Voutes, Greece}
\email{\mailto{nhatzitz@gmail.com}}
\urladdr{\url{http://www.nikoshatzizisis.wordpress.com/home/}}

\author[S. Kamvissis]{Spyridon Kamvissis $^{\ddag}$}
\address{$^{\ddag}$ Mathematics Building, University of Crete \\
700 13 Voutes, Greece}
\email{\mailto{spyros@tem.uoc.gr}}
\urladdr{\url{http://www.tem.uoc.gr/~spyros/}}

\thanks{}

\keywords{}

\subjclass[2000]{}

\bigskip

\begin{abstract}
In this paper we examine the semiclassical behavior of the scattering 
data of a 
non-self-adjoint Dirac operator with a fairly smooth -but not necessarily 
analytic- potential decaying at infinity. In particular, using ideas and 
methods going back to Langer and Olver, we provide a rigorous semiclassical 
analysis of the scattering coefficients, the Bohr-Sommerfeld condition for 
the location of the eigenvalues and their corresponding norming constants. 
Our analysis is motivated by the potential applications to the focusing 
cubic NLS equation, in view of the well-known fact discovered by Zakharov 
and Shabat that the spectral analysis of the Dirac operator is the basis 
of the solution of the NLS equation via inverse scattering theory. 
This paper complements and extends a previous work of Fujii\'e and the 
second author, which considered a more restricted problem for a strictly 
analytic potential.
\end{abstract}

\maketitle

\section{Introduction}

Consider the initial value problem (IVP) of the \textit{one-dimensional 
focusing nonlinear Schr\"odinger equation} (focusing NLS) for the complex 
field $u(x,t)$, i.e.
\be\label{ivp-nls}
\begin{cases}
i\hb\partial_t u+\frac{\hb^2}{2}\partial_x^2 u+|u|^2u=0,
\quad (x,t)\in\R\times\R\\\
u(x,0)=A(x),
\quad x\in\R
\end{cases}
\ee
for a real valued function $A$ and
a fixed positive number $\hb$.

\textit{Zakharov} and \textit{Shabat} in \cite{z+s_1972} have proved 
back in 1972 that the focusing
NLS equation is integrable via the \textit{Inverse Scattering Transform} 
(IST). A crucial step of the method is the analysis of the following 
\textit{Zakharov-Shabat (or Dirac) eigenvalue problem}
\be\label{zs-scattering}
\hb
\begin{bmatrix}
v_{1}'(x,\lam,\hb)\\
v_{2}'(x,\lam,\hb)
\end{bmatrix}=
\begin{bmatrix}
-i\lam & A(x) \\
-A(x) & i\lam
\end{bmatrix}
\begin{bmatrix}
v_{1}(x,\lam,\hb)\\
v_{2}(x,\lam,\hb)
\end{bmatrix}
\ee 
where $\lam\in\C$ is a “spectral parameter"; here prime denotes differentiation 
with respect to $x$.

Now let us suppose that $\hb$ is small compared to the $x,t$
we are interested in. The question raised is then: what is the behavior of solutions 
of the IVP (\ref{ivp-nls}) as $\hb\downarrow0$? 
The rigorous analysis of this problem was initiated in \cite{kmm}.
Because of the work of Zakharov and Shabat, 
the first step in the study of this IVP in the 
\textit{semiclassical limit}  $\hb\downarrow0$ has to be the \textit{asymptotic 
spectral analysis} of the \textit{scattering problem} (\ref{zs-scattering}) as 
$\hb\downarrow0$, keeping the function $A$ fixed.

The eigenvalue (EV) problem (\ref{zs-scattering}) is not and cannot 
be written as an EV
problem for a self-adjoint operator. What we  study here is a \textit{semiclassical
WKB problem} (or \textit{LG problem}) for the corresponding \textit{non-self-adjoint 
Dirac operator}
with \textit{potential} $A$. 

The question of the \textit{semiclassical approximation of the scattering data} has a deep significance
in view of the \textit{instability of the NLS problem} which appears in many levels.
In fact even away from the semiclassical regime, the  focusing
 NLS is the main model for the so-called
\say{\textit{modulational instability}} (as in \cite{bf}), 
although for positive fixed $\hb$ the initial value problem is well-posed.

Semiclassically the  instabilities become more pronounced.
One way to see this, is related to the underlying ellipticity
of the formal semiclassical limit.
To be more specific, consider the well-known \textit{Madelung transformation} 
\begin{equation}\nn
\begin{cases}
\rsf = |u|^2\\
\musf = \hb \im\,(u^* u_x)
\end{cases}
\end{equation}
where $u^*$ denotes the complex conjugate of $u$.
Then the IVP (\ref{ivp-nls}) becomes
\begin{equation}\nn
\begin{cases}
\rsf_t  +\musf_x = 0\\
\musf_t + \Big(\frac{\musf^2}{\rsf} + \frac{\rsf^2}{2}\Big)_x = 

\frac{\hb^2}{4} \partial_x[\rsf (\log \rsf)_{xx}]
\end{cases}
\end{equation}
with initial data $\rsf(x,0) =  |u|^2(x,0)=A^2(x)$ and $\musf(x,0)=0$.

The \textit{formal limit} as $\hb\downarrow0$ is
\begin{equation}\nn
\begin{cases}
\rsf_t  +\musf_x = 0\\
\musf_t + \Big(\frac{\musf^2}{\rsf} + \frac{\rsf^2}{2}\Big)_x = 0
\end{cases}
\end{equation}
with initial data $\rsf(x,0) =  |u|^2(x,0)=A^2(x)$ and $\musf(x,0)=0$.
This is an IVP for an \textit{elliptic system} of equations and so one expects that small 
perturbations of the initial data (independent of $\hb$) can lead to large changes in 
the solution, at any given time.

Instabilities appear also at different stages of the analysis:
the spectral analysis of the related non-self-adjoint Dirac operator,
the related \textit{equilibrium measure problem} (see \cite{kr}),  
the related \textit{Whitham equations} (cf. \cite{kmm}) which are also elliptic, 
the possibility of the appearance of \textit{rogue waves} (see \cite{bt})
and even the \textit{numerical studies} of the problem (as in \cite{mk}).

The semiclassical approximation of the scattering data results in
small changes of the initial data; changes $that ~depend~on~\hb$.
It is a priori unclear whether these small changes  can have a significant effect in the semiclassical 
asymptotics of the solution of the IVP (\ref{ivp-nls}) as $\hb\downarrow0$.
Our ultimate aim is to provide a proof that they do not.

Our work complements the
paper \cite{fujii+kamvi} of \textit{S. Fujii\'e} and the second author where
the potential is considered to be a \textit{real analytic bell-shaped function}
and in which the so-called \textit{exact WKB method} (cf. \cite{ecalle1984}, 
\cite{fln}, \cite{gg} and \cite{voros1983}) is employed.
In this work, we instead suppose that the bell-shaped potential function $A$ has 
\textit{only some prescribed smoothness} which we specify explicitly in \S\ref{assume}.
Our methods are necessarily different since the  exact WKB method requires analyticity.
Our ideas are rather influenced by the papers \cite{yafa2011} and \cite{yafa2018} of 
\textit{D. R. Yafaev} where an analogous problem is treated for the 
self-adjoint Schr\"odinger operator, which in turn rely on  
works \cite{olver1975} and \cite{olver1997} of \textit{F. W. J. Olver}
\footnote{Olver's work draws upon the studies of 
\textit{N. D. Kazarinoff, R. E. Langer} and \textit{R.W. McKelvey}
(see the references in \cite{olver1975}).}.
More precisely, Yafaev uses results mainly from \cite{olver1997} while we rely 
heavily on \cite{olver1975} as well. 

The present paper is arranged as follows. In section \S\ref{assume} we
state all the necessary assumptions on the potential function $A$ so that
Olver's work can be applied in our case. In section \S\ref{from-schrodi-to-dirac} 
we introduce a simple transformation that maps the Dirac problem to an equivalent Schr\"odinger
problem. Sections \S\ref{liouville-transform} and  \S\ref{continuity} show how the 
\textit{Liouville transformation} changes our Schr\"odinger equation
into one containing an \textit{error  term}  which is a continuous
function on the \textit{Liouville plane}.
By controlling this error term
in section \S\ref{approximate-solutions}  we obtain approximate 
solutions expressed with the help of \textit{Parabolic Cylinder Functions} (PCFs) 
in a \textit{Liouville variable} $\z\geq0$. 

In section \S\ref{ex} we illustrate the previously mentioned 
results for the special case where the potential function is 
$x\mapsto\frac{1}{1+x^2}$. Then in \S\ref{asympt-behave-sols} we
find the asymptotic behavior of the approximants introduced in 
\S\ref{approximate-solutions}.
In \S\ref{connection-formulas} we present some \textit{connection formulas}
that relate the approximate solutions for $\z\geq0$ to the ones for $\z\leq0$.
The significance of this connection becomes clear in \S\ref{quanta-evs}
where we find \textit{Bohr-Sommerfeld quantization conditions} for
the EVs of our problem, uniformly away from zero.

Next in \S\ref{near-zero-evs} we study the EVs that lie  closer to zero and
we are able to arrive at uniform  bounds all the way to zero. 
We only do this here for two specific (but quite inclusive) families 
of  functions $A$. It may be that more general conditions can alternatively
be posed instead; conditions
that would ensure the same results for a wide class of data.
But we are not able to do this at this point. Our situation is 
somewhat comparable to 
\cite{fujii+kamvi} where extra (general but complicated) conditions had to be 
added, in order to assure a good behavior of the EVs near $0$. 
Here however, the formula 
(\ref{psi-h}) we have for the function $\psi$,
makes it very easy to check if indeed the behavior of the EVs 
near zero, is good enough for any family of potentials $A$ defined 
by explicitly prescribed asymptotics at infinity. 

Section \S\ref{scattering} is concerned with the amplitudes of 
the transmission and reflection coefficients (both away and close to zero).
Unlike \cite{fujii+kamvi}, we do not need any extra conditions on $A$ to 
control the reflection coefficient (cf. \S\ref{scattering}) near $0$ 
(even though our estimates are somewhat weaker, they are still good enough 
for applications). Finally, we make some concluding remarks in section 
\S \ref{conclusion}.

For the sake of the reader, as the approximate solutions to our problems 
involve \textit{Airy} and Parabolic Cylinder Functions,
we present all the necessary results concerning these functions
in sections \ref{airy_functions} and \ref{parabolic-cylinder-functions} 
of the appendix. There, the reader can also find a section
(section \ref{exist-proof}) on a \textit{theorem concerning integral equations} 
which is the primary tool in the proof of the main Theorem \ref{main-thm}.

Notationwise, a bar over a letter (or number)
does \textit{not} denote complex conjugation. For complex conjugation 
we have reserved the superscript \say{$*$}; i.e. $z^*$ (and not $\bar{z}$) 
is the complex conjugate of $z$. The letter $C$ denotes generically 
a positive constant (appearing in estimates) and $\mathbb{R}_+$, 
$\mathbb{R}_-$ represent the sets of positive and negative numbers 
respectively. Also, we denote the \textit{Wronskian} 
of two functions $f$, $g$ by $\W[f,g]$. Furthermore, we write $f\sim g$
when $\tfrac{f}{g}$ tends to $1$. Finally, the notation $f^2(x)$ denotes 
the square of the value of the function $f$ at $x$. Hence, the symbols 
$f^2(x)$ and $f(x)^2$ are used interchangeably and are \textit{not} to 
be confused with the composition $f\circ f$ of $f$ with itself.

\section{The Potential}
\label{assume}

In this section we state precise assumptions on
the potential function $A$ which are sufficient to ensure that all the 
techniques and methods
developed in the following sections can go through easily. 
In short, we
consider bell-shaped functions with some smoothness.
To be more precise, we assume the following. 
\begin{assumption}
\label{primary-assume}
The function $A$ satisfies
\begin{itemize}
\item
$A(x)>0$ for $x\in\mathbb{R}$
\item
$A(-x)=A(x)$ for $x\in\mathbb{R}$
\item
$A$ is in $C^4(\R)$ and of class $C^5$ in a neighborhood of $0$
\item
$xA'(x)<0$ for $x\in\mathbb{R}\setminus\{0\}$
\item
$A''(0)<0$; we set $0<A(0)=:A_{max}$
\item
there exists $\tau>0$ so that
\begin{align*}
A(x)=\mathcal{O}\Big(\tfrac{1}{|x|^{1+\tau}}\Big)\quad\text{as}\quad
x\to\pm\infty
\\
A'(x)=\mathcal{O}\Big(\tfrac{1}{|x|^{2+\tau}}\Big)\quad\text{as}\quad
x\to\pm\infty
\\
A''(x)=\mathcal{O}\Big(\tfrac{1}{|x|^{3+\tau}}\Big)\quad\text{as}\quad
x\to\pm\infty.
\end{align*}
\end{itemize}
\end{assumption}

\begin{figure}[H]
\centering
\includegraphics[scale=0.35]{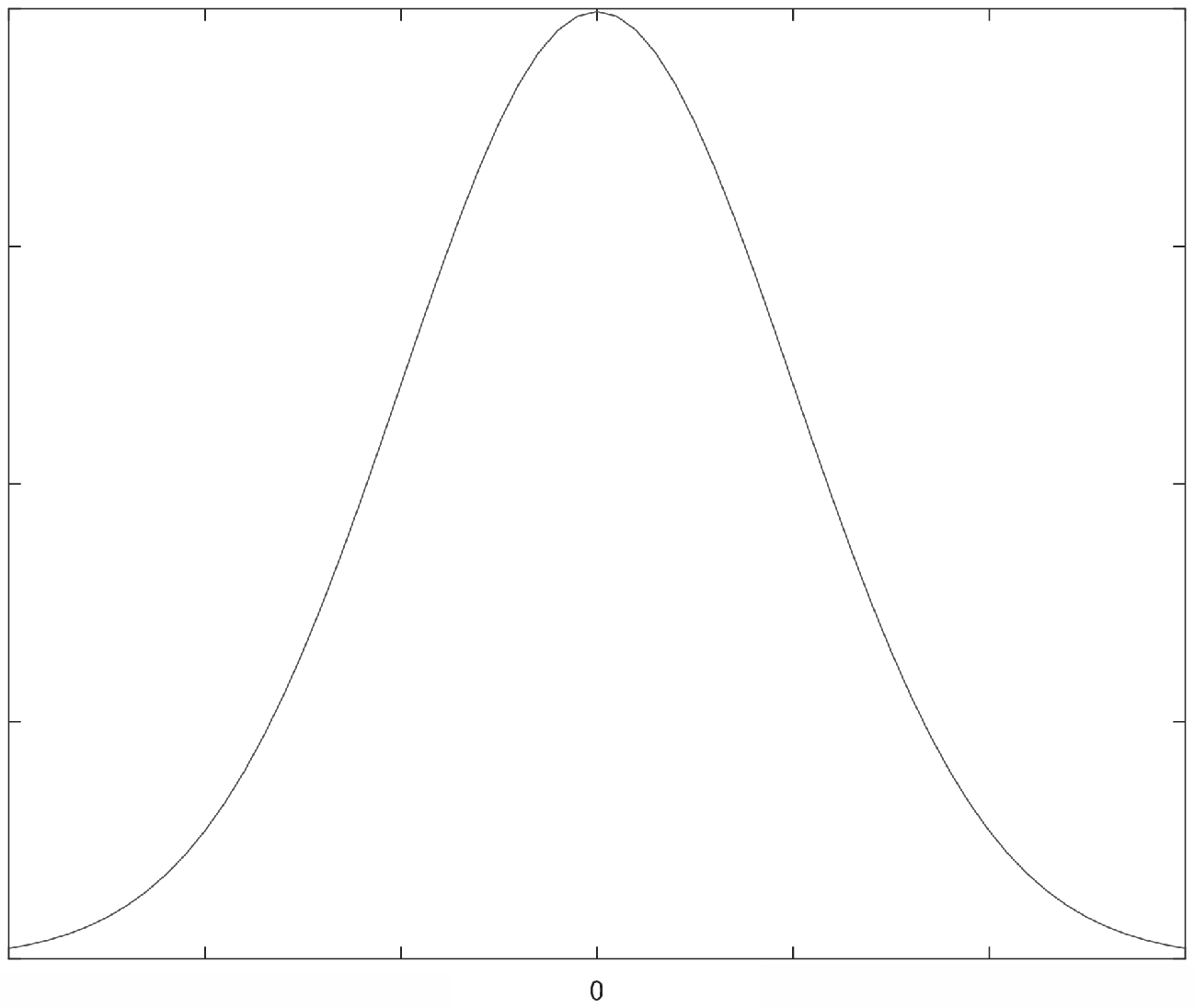}
\caption{A bell-shaped function.}
\end{figure}

Now let $\mu\in(0,A_{max}]\subset\mathbb{R}_+$. Observe that
\begin{itemize}
\item
for $\mu\in(0,A_{max})$ the equation $A(x)=\mu$ has exactly two solutions 
$x_{\pm}$ which of course depend on $\m$ and by the symmetry of $A$ satisfy 
$x_{\mp}=-x_{\pm}$. Furthermore, $A(x)>\mu$ for $x\in(x_{-},x_{+})$ 
and $A(x)<\mu$ for $x\in(-\infty,x_{-})\cup(x_{+},+\infty)$.
\item
when $\m=A_{max}$ the two points $x_\pm$ coalesce into one double at $x=0$.
\end{itemize}

We believe that neither the evenness assumption, nor the single local maximum 
assumption are strictly necessary. If evenness is not imposed, we have what 
\textit{Klaus} \& \textit{Shaw} call a “single lobe” potential 
(see \cite{klaus+shaw2002} and \cite{klaus+shaw2003}). Essentially our 
discussion in 
this article goes through mostly unaltered, since the results of Klaus \& Shaw 
mentioned below (in \S\ref{from-schrodi-to-dirac}) are still valid. 
One thing that changes is that the \textit{norming constants} 
are no more real, but we still have uniform estimates for them
(see Corollary \ref{norm-const}). If the single local maximum assumption 
is dropped we no more expect to have imaginary EVs. But we do expect to have 
EVs accumulate along curves as in sections \S\ref{quanta-evs}, 
\S\ref{near-zero-evs} and Bohr-Sommerfeld conditions to appear. 
A forthcoming paper will hopefully show how to handle more general cases.

\section{From Dirac to Schr\"odinger}
\label{from-schrodi-to-dirac}

As stated in the introduction, in this paper we examine the scattering
data for a Dirac operator. We start with an investigation of the EVs. 
Specifically, we study the EV problem
\be\label{ev-problem}
\mathfrak{D}_{\hbar}\mathbf{u}=\lambda\mathbf{u}
\ee
where $\mathfrak{D}_{\hbar}$ is the Dirac operator
\be\label{dirac}
\mathfrak{D}_{\hbar}:=
\begin{bmatrix}
i\hbar\partial_{x} & -iA \\
-iA & -i\hbar\partial_{x}
\end{bmatrix}
\ee
with $0<\hb\ll1$ a small parameter (\textit{Planck}), $A$ as in 
\S\ref{assume} and $\mathbf{u}=[u_{1}\hspace{3pt}u_{2}]^T$
is a function from $\mathbb{R}$ to $\mathbb{C}^2$. As usual, 
$\lambda\in\C$ plays the role of the spectral parameter. 

Let us make explicit what we mean when we discuss about the EVs of
the operator in (\ref{dirac}). We have the following definition. 
\begin{definition}
A number $\lambda\in\mathbb{C}$ is an \textbf{eigenvalue} of the operator 
$\mathfrak{D}_{\hbar}$ if equation (\ref{ev-problem}) has a 
non-trivial solution $\mathbf{u}\in L^{2}(\mathbb{R};\mathbb{C}^{2})$; 
that is
\begin{equation*}
0
<
\int_{-\infty}^{+\infty}\Big[|u_{1}(x)|^{2}+|u_{2}(x)|^{2}\Big]dx
<
+\infty.
\end{equation*}
\end{definition} 

The spectral characteristics of an operator like $\mathfrak{D}_{\hbar}$
having a potential $A$ satisfying the assumptions of \S\ref{assume} have 
been established in \cite{klaus+shaw2002} and \cite{klaus+shaw2003} 
by M. Klaus and J. K. Shaw. More precicely we know that
\begin{itemize}
\item
the continuous spectrum of $\mathfrak{D}_{\hbar}$ is the whole real 
line $\mathbb{R}$ and
\item
the EVs are simple, purely imaginary and symmetric with respect to 
the real axis; their imaginary part lying in $[-A_{max},A_{max}]$
\end{itemize}

The spectral facts above suggest writing $\lambda=i\mu$ for 
$0<\mu\leq A_{max}$ (due to symmetry). Hence, (\ref{ev-problem}) 
is written as
\be\label{z+sh}
\hbar
\begin{bmatrix}
u_{1}'(x,\mu,\hb)\\
u_{2}'(x,\mu,\hb)
\end{bmatrix}=
\begin{bmatrix}
\mu & A(x) \\
-A(x) & -\mu
\end{bmatrix}
\begin{bmatrix}
u_{1}(x,\mu,\hb)\\
u_{2}(x,\mu,\hb)
\end{bmatrix}
\ee 
where the prime denotes differentiation with respect to $x$.

Under the change of variables (cf. equation (4) in \cite{pdmiller2001})
\be\label{change-of-var-dirac-to-schrodi}
y_{\pm}=
\frac{u_{2}\pm u_{1}}{\sqrt{A\mp\mu}}
\ee
system (\ref{z+sh}) is equivalent to the following two independent 
equations
\be\label{schrodi_initial-version}
y_{\pm}''=
\bigg\{
\hbar^{-2}[\mu^2-A^2(x)]+
\frac{3}{4}\Big[\frac{A'(x)}{A(x)\mp\mu}\Big]^2-
\frac{1}{2}\frac{A''(x)}{A(x)\mp\mu}
\bigg\}
y_{\pm}
\ee
where we dropped the dependence of $y_{\pm}$ on $(x,\mu,\hb)$
for notational simplicity. Again, prime denotes differentiation 
with respect to $x$.

In (\ref{schrodi_initial-version}), we will only consider the 
equation with the lower index because  the term 
\be\nn
\frac{3}{4}\Big[\frac{A'(x)}{A(x)+\mu}\Big]^2-
\frac{1}{2}\frac{A''(x)}{A(x)+\mu}
\ee
has no singularities and  thus work with the equation
\be\label{schrodi_final-version}
\frac{d^2y}{dx^2}=
\bigg\{
\hbar^{-2}[\mu^2-A^2(x)]+
\frac{3}{4}\Big[\frac{A'(x)}{A(x)+\mu}\Big]^2-
\frac{1}{2}\frac{A''(x)}{A(x)+\mu}
\bigg\}
y.
\ee
Observe that the change of variables (\ref{change-of-var-dirac-to-schrodi}) 
with the \say{minus choice} does not alter the discrete spectrum.
Hence we are led to the following important fact.
\begin{proposition}
Finding the discrete spectrum of $\mathfrak{D}_{\hbar}$ in (\ref{dirac})
is equivalent to finding the values of  $\mu\in(0,A_{max}]$ for which
(\ref{schrodi_final-version}) has an $L^{2}(\mathbb{R})$ solution.
\end{proposition} 
 
Now, let us choose any $A_0$ such that $0<A_0<A_{max}$.
In equation (\ref{schrodi_final-version}) $x$ runs on the whole real 
line $\R$ and $\mu$ will play the role of a spectral parameter 
in $[A_{0},A_{max}]\subset\R_+$. For $\mu\in[A_{0},A_{max})$ the function 
$x\mapsto\mu^2-A^2(x)$ is non-vanishing on $\R$ except for two distinct 
simple zeros at $x=x_{-}$ and $x=x_{+}$ with $x_{-}<x_{+}$. In the
critical case $\mu=A_{max}$ the function $x\mapsto A_{max}^2-A^2(x)$ has
a single double zero at $x=0$. Both $x_{-}$, $x_{+}$ are continuous 
functions of the parameter $\mu$ and tend to zero as 
$\mu\uparrow A_{max}$.

\begin{figure}[H]
\centering
\begin{tikzpicture}

\begin{axis}[
    legend pos = north west,
    axis lines = none,
    xlabel = ,
    ylabel = {},
]
\addplot [
    domain=-4:4, 
    samples=100, 
    color=red,
]
{1/(1+x^2)};
\addlegendentry{$A(x)$}
\addplot [
    domain=-3:3, 
    samples=100, 
    color=blue,
    ]
    {1/2};

\addplot [
    domain=-3:3, 
    samples=100, 
    color=blue,
    ]
    {1/6};
    
\addplot [
    domain=-1:2, 
    samples=100, 
    color=blue,
    ]
    {1};    

\draw[scale=0.5,domain=-3:3,smooth,variable=\y,green]  plot ({6},{\y});

\draw[scale=0.5,domain=-0.05:0.25,dashed,variable=\y]  plot ({8.5},{\y});

\node[circle,inner sep=2pt,fill=black] at (1,.5) {};

\draw (-5,0) -- (5,0);

\draw (0,-1) -- (0,1.5);

\end{axis}

\draw (4.15,-0.5) node {$x=a$};

\draw (5.1,-0.2) node {$a_0$};

\draw (6,3) node {$y=\m$};

\draw (6,5.1) node {$A_{max}$};

\draw (6,1.1) node {$A_0$};

\end{tikzpicture}
\caption{The relationship between parameters $\m$ and $a$.}
\end{figure}
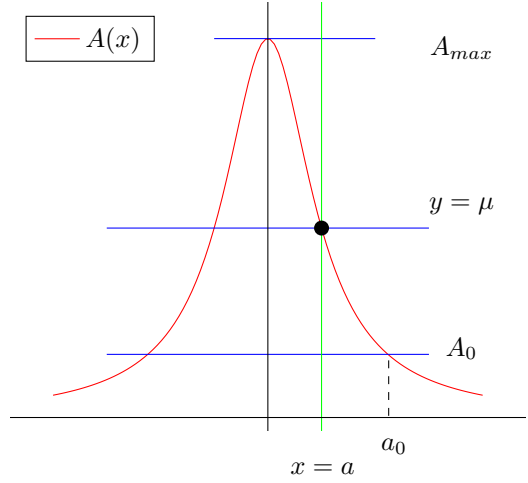

We introduce a change of variables for the spectral parameter $\mu$ in 
order to rely on the results from \cite{olver1975}. For this, 
we first define the function $B$ to be the restriction of $A$ on 
$[0,+\infty)$, i.e. $B=\restr{A}{[0,+\infty)}$ and note that, 
by the assumptions on $A$, the function $B$ is invertible. We set
\be\nn
a=x_{+}
\ee
so
\be\nn
A(x_{+})=\mu\Leftrightarrow
B(a)=\mu\Leftrightarrow
a=B^{-1}(\m).
\ee
Since $B$ is a decreasing function, so is $B^{-1}$ and we get
\be\nn
a\in B^{-1}\Big([A_{0},A_{max}]\Big)=[0,B^{-1}(A_{0})]=:[0,a_{0}].
\ee
Thus, the zeros of $x\mapsto A^{2}(a)-A^2(x)$ are located at
$x_{\pm}=\pm a$. Furthermore, the critical value of $a$ is now zero and
$a$ ranges over the compact interval $[0,a_{0}]$ (we should add that $A_{0}$ 
and consequently $a_{0}$ may be allowed to depend on $\hb$ as is the
case in \S\ref{near-zero-evs}).

With this new parameter, equation (\ref{schrodi_final-version}) 
is replaced by 
\be\label{final-schrodi}
\frac{d^2y}{dx^2}=[\hbar^{-2}f(x,a)+g(x,a)]y
\ee
where $f$ and $g$ satisfy
\be\label{v-new}
f(x,a)=A^2(a)-A^2(x)
\ee
and
\be\label{f-new}
g(x,a)=
\frac{3}{4}\Big[\frac{A'(x)}{A(x)+A(a)}\Big]^2-
\frac{1}{2}\frac{A''(x)}{A(x)+A(a)}.
\ee

We close this section with an important definition concerning
the zeros of $f$ in (\ref{v-new}).
\begin{definition}
The zeros (with respect to $x$) of the function 
\be\nn
f(x,a)=A^2(a)-A^2(x)
\ee
are called \textbf{turning ponts} (or \textbf{transition points}) of the
equation (\ref{final-schrodi}).
\end{definition}
Hence our equation facilitates two turning points at $x=\pm a$.

\section{The Liouville Transformation}
\label{liouville-transform}

In this section we introduce new variables $Y$ and $\z$ according to the 
Liouville transform
\be\nn
Y=\dot{x}^{-\frac{1}{2}}y
\ee
where the dot signifies differentiation with respect to $\z$. Equation 
(\ref{final-schrodi}) becomes
\be\label{schrodi-liouville-initial-form}
\frac{d^2Y}{d\z^2}=
\Big[\hbar^{-2}\dot{x}^2f(x,a)+\dot{x}^2g(x,a)+
\dot{x}^{\frac{1}{2}}\frac{d^2}{d\z^2}(\dot{x}^{-\frac{1}{2}})\Big]Y.
\ee
In our case, $f(\cdot,a)$ is negative in $(-a,a)$ and positive in
$(-\infty,-a)\cup(a,+\infty)$. Hence we prescribe
\be\label{our-case}
\dot{x}^2f(x,a)=\z^2-\alpha^2
\ee
where $\alpha\geq0$ is chosen in such a way that $x=-a$ corresponds to
$\z=-\alpha$ and $x=a$ to $\z=\alpha$ accordingly. Indeed, after integration,
(\ref{our-case}) yields
\be\label{our-case-integral-form}
\int_{-a}^{x}[-f(t,a)]^{\frac{1}{2}}dt=
\int_{-\alpha}^{\z}(\alpha^2-\tau^2)^{\frac{1}{2}}d\tau
\ee
provided that $-a\leq x\leq a$. Notice that by taking these integration 
limits, $-a$ corresponds to $-\alpha$. For the remaining correspondence, we 
require
\be\nn
\int_{-a}^{a}[-f(t,a)]^{\frac{1}{2}}dt=
\int_{-\alpha}^{\alpha}(\alpha^2-\tau^2)^{\frac{1}{2}}d\tau
\ee
and hence
\be\label{alpha}
\alpha^2=\frac{2}{\pi}\int_{-a}^{a}[-f(t,a)]^{\frac{1}{2}}dt.
\ee
For every fixed value of $\hb$, relation (\ref{alpha})  defines $\al$ as a 
continuous and increasing function of $a$ which vanishes as $a\downarrow0$
and equals $\al_0$ when $a=a_0$;  so $\al\in[0,\al_0]$.

Next, from (\ref{our-case-integral-form}) we find
\be\label{zeta-center}
\int_{-a}^{x}[-f(t,a)]^{\frac{1}{2}}dt=
\frac{1}{2}\al ^2\arccos\Big(-\frac{\z}{\al}\Big)+
\frac{1}{2}\z\big(\al ^2-\z ^2\big)^{\frac{1}{2}}
\quad\text{for}\quad
-a\leq x\leq a
\ee
with  the principal value choice for the inverse cosine taking values in 
$[0,\pi]$. For the remaining $x$-intervals, we integrate
(\ref{our-case}) to obtain
\be\label{zeta-left}
\int_{x}^{-a}[f(t,a)]^{\frac{1}{2}}dt=-
\frac{1}{2}\al ^2\arcosh\Big(-\frac{\z}{\al}\Big)-
\frac{1}{2}\z\big(\z ^2-\al ^2\big)^{\frac{1}{2}}
\quad\text{for}\quad
x\leq -a
\ee
and
\be\label{zeta-right}
\int_{a}^{x}[f(t,a)]^{\frac{1}{2}}dt=-
\frac{1}{2}\al ^2\arcosh\Big(\frac{\z}{\al}\Big)+
\frac{1}{2}\z\big(\z ^2-\al ^2\big)^{\frac{1}{2}}
\quad\text{for}\quad
x\geq a
\ee
with $\arcosh(x)=\ln\big(x+\sqrt{x^2-1}\big)$ for $x\geq1$.

Equations (\ref{zeta-center}), (\ref{zeta-left}) and (\ref{zeta-right}) 
show that $\z$ is a continuous and increasing function of $x$ in $\R$. 
Moreover, this shows that there is a one-to-one correspondence between 
these two variables. Finally, we substitute (\ref{our-case}) in 
(\ref{schrodi-liouville-initial-form}) and obtain
\be\label{schrodi-liouville-final-form}
\frac{d^2Y}{d\z^2}=
\big[\hb^{-2}(\z^2-\alpha^2)+\ps(\z,\al)\big]Y
\ee
where the \textit{error term} $\ps(\z,\al)$ is
\be\label{psi}
\ps(\z,\al)=
\dot{x}^2g(x,a)+
\dot{x}^{\frac{1}{2}}\frac{d^2}{d\z^2}(\dot{x}^{-\frac{1}{2}})
\ee
or equivalently
\begin{multline}\label{psi-equivalent}
\ps(\z,\al)=
\frac{1}{4}\frac{3\z^2 +2\al^2}{(\z^2 -\al^2)^2}
+\frac{1}{16}\frac{\z^2 -\al^2}{f^3(x,a)}
\Big\{4f(x,a)f''(x,a)-5[f'(x,a)]^2\Big\}\\
+(\z^2 -\al^2)\frac{g(x,a)}{f(x,a)}
\end{multline}
where prime denotes differentiation with respect to $x$.

In the critical case in which the two (simple) turning points coalesce into
one double point, we get a limit of the above transformation with $a=0$. So
\be\label{zeta-0-left}
\int_{x}^{0}[f(t,0)]^{\frac{1}{2}}dt=\frac{1}{2}\z^2
\quad\text{for}\quad
x\leq 0
\ee
\be\label{zeta-0-right}
\int_{0}^{x}[f(t,0)]^{\frac{1}{2}}dt=\frac{1}{2}\z^2
\quad\text{for}\quad
x\geq 0
\ee
and equations (\ref{schrodi-liouville-final-form}), (\ref{psi}) and
(\ref{psi-equivalent}) apply with $a=\al=0$.

\section{Two useful lemmas}
\label{continuity}

In this section we prove two helpful assertions that will be facilitated
in the following sections. First, that the error function $\ps(\z,\al)$ 
resulting from the Liouville transformation of \S \ref{liouville-transform}, 
is continuous in $\al$ and $\z$; a fact that will be used in 
\S \ref{approximate-solutions} to prove the existence of approximate 
solutions of equation (\ref{schrodi-liouville-final-form}). 
Secondly, we give asymptotics of $x$ for big values of $\z$.

The first lemma that concerns the error term in 
(\ref{schrodi-liouville-final-form}) is the following.
\begin{lemma}
\label{junk-continuity}
The function $\psi(\z,\alpha)$ in equation (\ref{schrodi-liouville-final-form}) 
as defined in (\ref{psi}) is continuous in $\z$ and $\alpha$ in the region
$(-\infty,+\infty)\times[0,\alpha_0]$ of the $(\z,\alpha)$-plane.
\end{lemma}
\begin{proof}
First we introduce an auxilliary function $p$. We define it by
\be\label{p}
f(x,a)=(x^2-a^2)p(x,a)
\ee
where 
\be\nn
p(\pm a,a)=\mp\frac{A(a)A'(\pm a)}{a}>0
\quad\text{for}\quad 
a\in(0,a_{0}] 
\ee
and
\be\nn
p(0,0)=-A_{max}A''(0)>0.
\ee
The functions $f$, $g$ and $p$ defined by (\ref{v-new}), (\ref{f-new}) and 
(\ref{p}) respectively satisfy the following properties
\begin{itemize}
\item[(i)]
$p, \frac{\partial p}{\partial x}, \frac{\partial^2 p}{\partial x^2}$ and
$g$ are continuous functions of $x$ and $a$ (this means in $x$ and $a$
simultaneously and not separately) in the region $\R\times[0,a_{0}]$ 
\item[(ii)]
$p$ is positive throughout the same region
\item[(iii)]
$|\frac{\partial^3 p}{\partial x^3}|$ is bounded in a neighborhood of the point
$(x,a)=(0,0)$ in the same region and
\item[(iv)]
$f$ is a non-increasing function of $a$ when $x\in[-a,a]$ and $a\in[0,a_{0}]$.
\end{itemize}
Indeed, (i) and (iii) follow from (\ref{p}) and the fact that $A$ is 
in $C^4$ and of class $C^5$ 
in some neighborhood of $0$ (see \S\ref{assume}). For (ii) recall the sign of 
$f$ using (\ref{v-new}). Finally (iv) is a consequence of the monotonicity of 
$A$ in $[0,+\infty)$ (again cf. \S\ref{assume}). By Lemma I  in Olver's paper 
\cite{olver1975}, the function 
$\ps$ defined by (\ref{psi}) (or (\ref{psi-equivalent})) is continuous in 
the corresponding region of the $(\z,\al)$-plane.
\end{proof}

Finally, recall (\ref{zeta-right}). It shows that $x\uparrow+\infty$ as 
$\z\uparrow+\infty$. The lemma below deals with the asymptotic behavior 
of $x$ as $\z\uparrow+\infty$. 
\begin{lemma}
\label{x-at-big-zeta}
Considering $x$ as a function of $\z$ we see that
\be
\label{x-asymptotics}
x=\frac{\z^2}{2A(a)}\Big[1+\asympt\Big(\tfrac{\log\z}{\z^2}\Big)\Big]
\quad
\text{as}
\hspace{7pt}
\z\uparrow+\infty
\ee
uniformly with respect to $a\in[0,a_0]$.
\end{lemma}
\begin{proof}
By (\ref{v-new}) and (\ref{our-case}) we have
\be\nn
\Big(\frac{dx}{d\z}\Big)^2[A^2(a)-A^2(x)]=
\begin{cases}
\alpha^2-\z^2,\quad0\leq\z\leq\al\\
\z^2-\alpha^2,\quad\z>\al
\end{cases}
\ee
Choosing $x_0$ to satisfy $\z(x_0)=0$, we have
\be\nn
\int_{x_0}^{x}dt=
\frac{1}{A(a)}
\bigg(
\int_{0}^{\alpha}\sqrt{\alpha^2-\eta^2}d\eta+
\int_{\alpha}^{\z}\sqrt{\eta^2-\alpha^2}d\eta
\bigg)
\quad
\text{as}\hspace{7pt}
\z\uparrow+\infty.
\ee
We obtain
\be\nn
x-x_0=
\frac{\pi\alpha^2}{4A(a)}+
\frac{\z^2}{2A(a)}\Big[1+\asympt\Big(\tfrac{\log\z}{\z^2}\Big)\Big]
\quad
\text{as}\hspace{7pt}
\z\uparrow+\infty
\ee
from which the desired result follows.
\end{proof}

\section{Approximate Solutions}
\label{approximate-solutions}

In this section we exploit the tools assembled in the previous sections.
Here, we state a theorem concerning approximate solutions of equation (\ref{schrodi-liouville-final-form}), i.e.
\be
\label{liouville-final-form}
\frac{d^2Y}{d\z^2}=
\big[\hb^{-2}(\z^2-\alpha^2)+\ps(\z,\al)\big]Y
\ee
where the error term $\ps(\z,\al)$ is
\begin{multline}\label{psi-final-form}
\ps(\z,\al)=
\frac{1}{4}\frac{3\z^2 +2\al^2}{(\z^2 -\al^2)^2}
+\frac{1}{16}\frac{\z^2 -\al^2}{f^3(x,a)}
\Big\{4f(x,a)f''(x,a)-5[f'(x,a)]^2\Big\}\\
+(\z^2 -\al^2)\frac{g(x,a)}{f(x,a)}.
\end{multline}
To this goal, we need a way to assess the error. So we introduce an 
\textit{error-control function} $H$ along with a 
\textit{balancing function} $\Om$
\footnote{For $\Om$ we can actually choose any continuous function of 
the real variable $x$ which is positive (except possibly at $x=0$) and 
satisfies the asymptotics
$\Om(x)=\asympt(|x|^{\frac{1}{3}})
\quad\text{as}\quad
x\to\pm\infty$.} 
.
\begin{definition}
Define the function $\Om$ by
\be\label{omega-asymptotics}
\Om(x)=1+|x|^{\frac{1}{3}}.
\ee 
As an \textbf{error-control function} $H(\z,\al,\hb)$ of equation 
(\ref{liouville-final-form}) we consider any primitive of
the function
\be\nn
\frac{\ps(\z,\al)}{\Om(\z\sqrt{2\hb^{-1}})}.
\ee
\end{definition} 

Furthermore, we need the notion of the \textit{variation} of 
the error-control function $H$. We have
\begin{definition}
Take $\z_1<\z_2$ for $\z_1\in[0,+\infty)$ and 
$\z_2\in(0,+\infty)\cup\{+\infty\}$. The \textbf{variation} $\var_{\z_1,\z_2}[H]$ 
in the interval $(\z_1,\z_2)$ of the error-control function $H$ 
of equation (\ref{liouville-final-form}) is defined by
\be\nn
\var_{\z_1,\z_2}[H](\al,\hb)=
\int_{\z_1}^{\z_2}\frac{|\ps(t,\al)|}{\Om(t\sqrt{2\hb^{-1}})}dt.
\ee
\end{definition}

Before stating the main theorem, we also need to define an
auxiliary function that shows up in the error estimates of the
approximate solutions. So for any $b\leq0$ we set
\be\label{lambda-function}
l(b):=\sup_{x\in(0,+\infty)}
\bigg\{\Om(x)\frac{\msf(x,b)^2}{\G(\tfrac{1}{2}-b)}\bigg\}
\ee
where $\msf$ is a function defined in terms of Parabolic Cylinder 
Functions in section 
\ref{parabolic-cylinder-functions} of the appendix and $\Gamma$ denotes the
\textit{Gamma function}.
We note that the above supremum is finite for each value of $b$. This fact 
is a consequence of (\ref{omega-asymptotics}) and the first relation in
(\ref{M,N-asymptotics}). Furthermore, because the relations 
(\ref{M,N-asymptotics}) hold uniformly in compact intervals of the parameter 
$b$, the function $l$ is continuous.

We are now ready for the  main theorem of this section.
\btheo\label{main-thm}
For each value of $\hb$, the equation (\ref{liouville-final-form})
has in the region $[0,+\infty)\times[0,\al_{0}]$ of the $(\z,\al)$-plane   
two solutions $Y_1$ and $Y_2$ satisfying
\bea
\label{y1-approx}
Y_1(\z,\al,\hb)=U(\z\sqrt{2\hb^{-1}},-\tfrac{1}{2}\hb^{-1}\al^2)+
\eps_1 (\z,\al,\hb)\\
\label{y2-approx}
Y_2(\z,\al,\hb)=\ol U(\z\sqrt{2\hb^{-1}},-\tfrac{1}{2}\hb^{-1}\al^2)+
\eps_2 (\z,\al,\hb)
\eea
where $U$, $\ol U$ are the PCFs defined in appendix 
\ref{parabolic-cylinder-functions}. These two solutions $Y_1$, $Y_2$ 
are continuous and have continuous first and second partial 
$\z$-derivatives. The errors $\eps_1$, $\eps_2$ in the relations 
above satisfy the estimates
\begin{multline}\label{junk1}
\frac{|\eps_1 (\z,\al,\hb)|}{\msf(\z\sqrt{2\hb^{-1}},-\tfrac{1}{2}\hb^{-1}\al^2)},
\frac{\Big|\frac{\partial \eps_1 }{\partial\z}(\z,\al,\hb)\Big|}{\sqrt{2\hb^{-1}}
\nsf(\z\sqrt{2\hb^{-1}},-\tfrac{1}{2}\hb^{-1}\al^2)}\\
\leq
\frac{1}{\esf(\z\sqrt{2\hb^{-1}},-\tfrac{1}{2}\hb^{-1}\al^2)}
\Big(\exp\big\{\tfrac{1}{2}(\pi\hb)^{\frac{1}{2}}l(-\tfrac{1}{2}\hb^{-1}\al^2)
\mathcal{V}_{\z,+\infty}[H](\al,\hb)\big\}-1\Big)
\end{multline}
and
\begin{multline}\label{junk2}
\frac{|\eps_2 (\z,\al,\hb)|}{\msf(\z\sqrt{2\hb^{-1}},-\tfrac{1}{2}\hb^{-1}\al^2)},
\frac{\Big|\frac{\partial \eps_2 }{\partial\z}(\z,\al,\hb)\Big|}{\sqrt{2\hb^{-1}}
\nsf(\z\sqrt{2\hb^{-1}},-\tfrac{1}{2}\hb^{-1}\al^2)}\\
\leq
\esf(\z\sqrt{2\hb^{-1}},-\tfrac{1}{2}\hb^{-1}\al^2)
\Big(\exp\big\{\tfrac{1}{2}(\pi\hb)^{\frac{1}{2}}l(-\tfrac{1}{2}\hb^{-1}\al^2)
\mathcal{V}_{0,\z}[H](\al,\hb)\big\}-1\Big)
\end{multline}
\etheo
\begin{proof}
By Theorem \ref{thm-on-exist-int-eq} (cf. Theorem I in \cite{olver1975}),
it suffices to prove that
\begin{itemize}
\item
the function $\ps$ is continuous in the region 
$[0,+\infty)\times[0,\al_{0}]$ and
\item
the integral
\be
\label{variation-total}
\var_{0,+\infty}[H](\al,\hb)=
\int_{0}^{+\infty}\frac{|\ps(t,\al)|}{\Om(t\sqrt{2\hb^{-1}})}dt
\ee 
converges  uniformly in $\al$.
\end{itemize}
The first assertion  has already been 
proven in Lemma \ref{junk-continuity}. For the second, we argue as follows.
Using (\ref{v-new}), (\ref{f-new}) and (\ref{psi-equivalent}) we find
\begin{align*}
\frac{\ps(\z,\al)}{\z^{1/3}}
& =
\frac{1}{4}\frac{3\z^2 +2\al^2}{\z^{1/3}(\z^2 -\al^2)^2}\\
& +
\frac{1}{16}\frac{\z^2 -\al^2}{\z^{1/3}[A^2(a)-A^2(x)]^3}\cdot\\
& \hspace{1.2cm}
\big\{
-8A^2(a)[A'(x)^2+A(x)A''(x)]
-12A^2(x)A'(x)^2
+8A^3(x)A''(x)
\big\}\\
& +
\frac{\z^2 -\al^2}{\z^{1/3}[A^2(a)-A^2(x)]}
\Bigg\{\frac{3}{4}\Big[\frac{A'(x)}{A(x)+A(a)}\Big]^2-
\frac{1}{2}\frac{A''(x)}{A(x)+A(a)}\Bigg\}.
\end{align*}
This in addition to the asymptotics for $A$ in \S\ref{assume} 
and (\ref{x-asymptotics}), implies that 
$|\ps(\z,\al)|/\z^{\frac{1}{3}}$ is integrable
at $\z=+\infty$ and hence the variation (\ref{variation-total}) 
is finite; in fact uniformly bounded in  $\al$.
\end{proof}

\section{An Example}
\label{ex}

In this section we illustrate the theory developed so far to the
special case of the potential $A(x)=\tfrac{1}{1+x^2}$, $x\in\R$. 
First, observe that this particular potential $A$ satisfies the 
assumptions of \S\ref{assume}; indeed 
\begin{itemize}
\item
it is always positive, even and smooth,
\item
it is increasing in $(-\infty,0]$ and decreasing in $[0,+\infty)$,
\item
it has a maximum at $x=0$, namely $A_{max}=A(0)=1$,
\item
$\|A\|_{L^1(\R)}=\pi$ and
\item
if $\m\in(0,1)$ the equation $A(x)=\m$ gives the two zeros 
$x_{\pm}=\pm\sqrt{\m^{-1}-1}$ while for $\m=1$ we get a double solution
$x=0$.
\end{itemize}

When $\m\in[A_0,1]$ for $A_0>0$ ($\m=1$ corresponding to the critical case), 
the parameter 
$a=x_{+}=\sqrt{\m^{-1}-1}$ ranges over $[0,a_0]$ where $a_0=\sqrt{A_0^{-1}-1}$ 
(the criticality now being $a=0$). The equation in question is
\be\label{example-equation}
\frac{d^2y}{dx^2}=[\hbar^{-2}f(x,a)+g(x,a)]y
\ee
where $f$ and $g$ satisfy
\be\label{example-f}
f(x,a)=A^2(a)-A^2(x)=\frac{(x^2-a^2)(x^2+a^2+2)}{[(1+a^2)(1+x^2)]^2}
\ee
and
\begin{align}\label{example-g}
\nn
g(x,a) & =
\frac{3}{4}\Big[\frac{A'(x)}{A(x)+A(a)}\Big]^2-
\frac{1}{2}\frac{A''(x)}{A(x)+A(a)}\\
& =
\frac{(1+a^2)(-3x^4-2x^2+a^2+2)}{[(1+x^2)(x^2+a^2+2)]^2}.
\end{align}
The function $p$ that satisfies $f(x,a)=(x^2-a^2)p(x,a)$ is
\be\label{example-p}
p(x,a)=\frac{x^2+a^2+2}{[(1+a^2)(1+x^2)]^2}.
\ee

For the non-critical case, applying the Liouville transform 
\be\nn
Y=\dot{x}^{-\frac{1}{2}}y,\quad\dot{x}^2f(x,a)=\z^2-\alpha^2,
\ee
where $\al\in(0,\al_0]$ in which
$\al_0>0$ satisfies
\begin{align*}
\al_0^2 & = 
\frac{2}{\pi}\int_{-a_0}^{a_0}[-f(t,a_0)]^{\frac{1}{2}}dt\\
& =
\frac{4}{\pi(1+a_0^2)}\int_{0}^{a_0}
\frac{\sqrt{(a_0^2-t^2)(t^2+a_0^2+2)}}{1+t^2}dt,
\end{align*}
(cf. (\ref{alpha})) we get
\begin{multline*}
\frac{1}{1+a^2}\int_{x}^{-a}\frac{\sqrt{(t^2-a^2)(t^2+a^2+2)}}{1+t^2}dt=\\-
\frac{1}{2}\al ^2\arcosh\Big(-\frac{\z}{\al}\Big)-
\frac{1}{2}\z\big(\z ^2-\al ^2\big)^{\frac{1}{2}}
\quad\text{for}\quad
x\leq -a
\end{multline*}
(cf. (\ref{zeta-left})) and 
\begin{multline*}
\frac{1}{1+a^2}\int_{-a}^{x}\frac{\sqrt{(a^2-t^2)(t^2+a^2+2)}}{1+t^2}dt=\\
\frac{1}{2}\al ^2\arccos\Big(-\frac{\z}{\al}\Big)+
\frac{1}{2}\z\big(\al ^2-\z ^2\big)^{\frac{1}{2}}
\quad\text{for}\quad
-a\leq x\leq a
\end{multline*}
(cf. (\ref{zeta-center})),
in which the inverse cosine takes its principal value (i.e. the value in $[0,\pi]$) 
and
\begin{multline*}
\frac{1}{1+a^2}\int_{a}^{x}\frac{\sqrt{(t^2-a^2)(t^2+a^2+2)}}{1+t^2}dt=\\-
\frac{1}{2}\al ^2\arcosh\Big(\frac{\z}{\al}\Big)+
\frac{1}{2}\z\big(\z ^2-\al ^2\big)^{\frac{1}{2}}
\quad\text{for}\quad
x\geq a
\end{multline*}
(cf. (\ref{zeta-right})). Additionally, equation (\ref{example-equation}) 
is transformed to
\be\nn
\frac{d^2Y}{d\z^2}=
\big[\hb^{-2}(\z^2-\alpha^2)+\ps(\z,\al)\big]Y
\ee
where
\begin{multline}
\label{psi-example}
\ps(\z,\al)=
\frac{1}{4}\frac{3\z^2 +2\al^2}{(\z^2 -\al^2)^2}
-(1+a^2)^4(\z^2 -\al^2)\frac{5x^6+9x^4+3x^2+a^4+2a^2}{[(x^2-a^2)(x^2+a^2+2)]^3}\\
+(1+a^2)^3(\z^2 -\al^2)\frac{-3x^4-2x^2+a^2+2}{(x^2-a^2)(x^2+a^2+2)^3}.
\end{multline}

In the critical case ($a=\al=0$) we have
\be\nn
\int_{0}^{x}\frac{t\sqrt{2+t^2}}{1+t^2}dt=\frac{1}{2}\z^2
\quad\text{for}\quad
x\in\R
\ee
(cf. (\ref{zeta-0-left}) and (\ref{zeta-0-right})) and
\be\label{psi-example-critical}
\ps(\z,0)=
\frac{3}{4}\frac{1}{\z^2}-\z^2\frac{3x^6+7x^4+7x^2+3}{x^4(x^2+2)^3}.
\ee
We  note that all the integrals above can be explicitly evaluated in 
terms of  elliptic integrals.

From  (\ref{example-p}) we have
\begin{align*}
\frac{\partial p}{\partial x}(x,a) & =
\frac{2x(-x^2-2a^2-3)}{(1+a^2)^2(1+x^2)^3}
\\
\frac{\partial^2 p}{\partial x^2}(x,a) & =
\frac{2(3x^4+10a^2x^2+12x^2-2a^2-3)}{(1+a^2)^2(1+x^2)^4}
\\
\frac{\partial^3 p}{\partial x^3}(x,a) & =
\frac{24x(-x^4-5a^2x^2-5x^2+3a^2+4)}{(1+a^2)^2(1+x^2)^5}
\end{align*}
and from (\ref{example-f}) we have
\be\nn
\frac{\partial f}{\partial a}(x,a)=-\frac{4a}{(1+a^2)^3}.
\ee
Hence, these last observations about $f$ and $p$ along with 
(\ref{example-g}) clearly show that 
\begin{itemize}
\item[(i)]
$p, \frac{\partial p}{\partial x}, \frac{\partial^2 p}{\partial x^2}$ and
$g$ are continuous functions in the region $\R\times[0,a_0]$
\item[(ii)]
$p$ is positive in $\R\times[0,a_0]$
\item[(iii)]
$|\frac{\partial^3 p}{\partial x^3}|(0,0)=0$ and
\item[(iv)]
$\frac{\partial f}{\partial a}<0$ when $x\in[-a,a]$ and $a\in(0,a_0]$.
\end{itemize}
As argued in \S\ref{continuity}, these four properties imply that  $\ps$ 
is continuous in the region $\R\times[0,\al_0]$ of the $(\z,\al)$-plane.

Now, the variation integral
\be\nn
\int_0^{+\infty}\frac{|\ps(t,\al)|}{\Om(t\sqrt{2\hb^{-1}})}dt
\ee
where $\ps$ is given by (\ref{psi-example}) or (\ref{psi-example-critical}) 
and $\Om(x)=1+|x|^{\frac{1}{3}}$, converges uniformly for $\al\in[0,\al_0]$ 
as $\hb\downarrow0$. So we can obtain the two specific approximate solutions 
guaranteed by Theorem \ref{main-thm}.

\section{Asymptotic Behavior of Solutions}
\label{asympt-behave-sols}

In order to deduce the asymptotic behavior of the solutions 
$Y_1(\z,\al,\hb), Y_2(\z,\al,\hb)$ when $\hb\downarrow0$, we need to 
determine the asymptotic form as $\hb\downarrow0$ of the error bounds
(\ref{junk1}), (\ref{junk2}) examining closely 
$l(-\tfrac{1}{2}\hb^{-1}\al^2)$ and $\var_{0,+\infty}[H](\al,\hb)$
for $\alpha\in[0,\al_0]$.

We start 
\footnote{
The subsequent analysis follows the idea found in \S 6.2 of \cite{olver1975}.
}
by investigating $l(b)$ as in (\ref{lambda-function}) for 
$b\downarrow-\infty$. Take $\nu\geq1$ to be a large positive number 
and set $b=-\frac{1}{2}\nu^2$ and $x=\nu y\sqrt{2}$. Then by (\ref{m-definition}), 
(\ref{u-asymptotics}) and (\ref{ubar-asymptotics}) the quantity
\be\label{m2-dia-gamma}
\frac{\msf(\nu y\sqrt{2},-\frac{1}{2}\nu^2)^2}{\G(\tfrac{1}{2}+\frac{1}{2}\nu^2)}
\ee
is equal to
\be\nn
\sqrt{\frac{16\pi\nu^{-\frac{2}{3}}\eta}{y^2-1}}\cdot
\begin{cases}
\Big[\Ai^2(\nu^{\frac{4}{3}}\eta)+\Bi^2(\nu^{\frac{4}{3}}\eta)+
E^2(\nu^{\frac{4}{3}}\eta)M^2(\nu^{\frac{4}{3}}\eta)\asympt(\nu^{-2})\Big],
0\leq y\leq\tfrac{\ro(-\frac{1}{2}\nu^2)}{\nu\sqrt{2}}\\
\\
\Big[\Ai(\nu^{\frac{4}{3}}\eta)\Bi(\nu^{\frac{4}{3}}\eta)+
M^2(\nu^{\frac{4}{3}}\eta)\asympt(\nu^{-2})\Big],
y\geq\tfrac{\ro(-\frac{1}{2}\nu^2)}{\nu\sqrt{2}}
\end{cases}
\ee
where in each case, the estimate $\asympt(\nu^{-2})$ is uniform with 
respect to $y$. Using (\ref{largest-root-pcf-asympt}) we see that  
\be\nn
\frac{\ro(-\frac{1}{2}\nu^2)}{\nu\sqrt{2}}=
1+2^{-\frac{1}{3}}c_*\nu^{-\frac{4}{3}}+\asympt(\nu^{-\frac{8}{3}})
\quad\text{as}\quad\nu\to+\infty
\ee
and consequently $\nu^{\frac{4}{3}}\eta=c_*+\asympt(\nu^{-\frac{4}{3}})$.
But $E$ is bounded in $[0,c_*+\asympt(\nu^{-\frac{4}{3}})]$. Hence we may 
write (\ref{m2-dia-gamma}) as
\be\nn
\frac{4\sqrt{\pi}}{\nu^{\frac{1}{3}}}
\Big(\frac{\eta}{y^2-1}\Big)^{\frac{1}{2}}\cdot
\begin{cases}
\Big[\Ai^2(\nu^{\frac{4}{3}}\eta)+\Bi^2(\nu^{\frac{4}{3}}\eta)+
M^2(\nu^{\frac{4}{3}}\eta)\asympt(\nu^{-2})\Big],
\hspace{3pt} 0\leq y\leq\frac{\ro(-\frac{1}{2}\nu^2)}{\nu\sqrt{2}}\\
\\
\Big[\Ai(\nu^{\frac{4}{3}}\eta)\Bi(\nu^{\frac{4}{3}}\eta)+
M^2(\nu^{\frac{4}{3}}\eta)\asympt(\nu^{-2})\Big],
\hspace{3pt} y\geq\frac{\ro(-\frac{1}{2}\nu^2)}{\nu\sqrt{2}}
\end{cases}
\ee
where the $\asympt$-terms are again uniform in $y$.

Next, we employ the asymptotic approximations for the functions $\Ai, \Bi$
and $M$ (cf. section \ref{airy_functions} in appendix) so that for $y\geq1$ we
obtain 
\be\label{m2/g-1}
\frac{\msf(\nu y\sqrt{2},-\frac{1}{2}\nu^2)^2}{\G(\tfrac{1}{2}+\frac{1}{2}\nu^2)}\leq
\frac{C}{\nu^{\frac{1}{3}}}
\Big(\frac{\eta}{y^2-1}\Big)^{\frac{1}{2}}
\frac{1}{1+\nu^{\frac{2}{3}}\eta^{\frac{1}{2}}}
\ee
where $C$ denotes a positive constant, used generically in what follows.
By (\ref{eta-definition}) we have 
$\eta\sim(\tfrac{3}{4})^{\frac{2}{3}}y^{\frac{4}{3}}$ as $y\to+\infty$, whence
for $y\geq0$ the estimate
\be\label{m2/g-2}
\Big(\frac{\eta}{y^2-1}\Big)^{\frac{1}{2}}\leq
\frac{C}{1+\eta^{\frac{1}{4}}}.
\ee
Also, from  (\ref{omega-asymptotics}) we have
\be\label{m2/g-3}
\Om(\nu y\sqrt{2})\leq 
C(1+\nu^{\frac{1}{3}}y^{\frac{1}{3}})\leq
C\nu^{\frac{1}{3}}(1+\eta^{\frac{1}{4}}).
\ee
Finally, combining (\ref{m2/g-1}), (\ref{m2/g-2}) and (\ref{m2/g-3}) we get
\be\nn
\Om(\nu y\sqrt{2})
\frac{\msf(\nu y\sqrt{2},-\frac{1}{2}\nu^2)^2}{\G(\tfrac{1}{2}+\frac{1}{2}\nu^2)}
\leq
C
\frac{1}{1+\nu^{\frac{2}{3}}\eta^{\frac{1}{2}}}\leq
C
\ee
implying that 
\be\label{lambda-asympt}
l(-\tfrac{1}{2}\hb^{-1}\al^2)=\asympt(1)
\quad\text{as}\quad
\hb\downarrow0.
\ee

Next, we examine $\var_{0,+\infty}[H](\al,\hb)$. In the proof
of Theorem \ref{main-thm} we showed that $|\ps(\z,\al)|/\z^{\frac{1}{3}}$ 
is integrable at $\z=+\infty$ uniformly with respect to $\al$.
Thus
\begin{align}\label{v-operator-asympt}
\var_{0,+\infty}[H](\al,\hb)\leq
C\int_{0}^{1} & \frac{dt}{1+(t\sqrt{2\hb^{-1}})^{\frac{1}{3}}}\\ \nn
& +
\Big(\frac{\hb}{2}\Big)^{\frac{1}{6}}
\int_{1}^{+\infty}\frac{|\ps(t,\al)|}{t^{\frac{1}{3}}}dt=
\asympt(\hb^{1/6})
\quad\text{as}\quad
\hb\downarrow0.
\end{align}

The last two relations applied to (\ref{junk1}) and (\ref{junk2}) supply us 
with the required results
\footnote{
Observe that since $\ps(\z,\al)$ is integrable at $\z=+\infty$
the same results could be achieved by demanding $\Om(x)=1$ for all $x$.
We chose to present the general case since it is more broadly applicable.}
\begin{align}
\label{e1-asympt}
\eps_1(\z,\al,\hb) & =
\frac{\msf(\z\sqrt{2\hb^{-1}},-\tfrac{1}{2}\hb^{-1}\al^2)}
{\esf(\z\sqrt{2\hb^{-1}},-\tfrac{1}{2}\hb^{-1}\al^2)}
\asympt(\hb^{\frac{2}{3}})
\\
\nn
\eps_2(\z,\al,\hb) & =
\esf(\z\sqrt{2\hb^{-1}},-\tfrac{1}{2}\hb^{-1}\al^2)
\msf(\z\sqrt{2\hb^{-1}},-\tfrac{1}{2}\hb^{-1}\al^2)
\asympt(\hb^{\frac{2}{3}})
\\
\nn
\frac{\partial \eps_1 }{\partial\z}(\z,\al,\hb) & =
\frac{\nsf(\z\sqrt{2\hb^{-1}},-\tfrac{1}{2}\hb^{-1}\al^2)}
{\esf(\z\sqrt{2\hb^{-1}},-\tfrac{1}{2}\hb^{-1}\al^2)}
\asympt(\hb^{\frac{1}{6}})
\\
\nn
\frac{\partial \eps_2 }{\partial\z}(\z,\al,\hb) & =
\esf(\z\sqrt{2\hb^{-1}},-\tfrac{1}{2}\hb^{-1}\al^2)
\nsf(\z\sqrt{2\hb^{-1}},-\tfrac{1}{2}\hb^{-1}\al^2)
\asympt(\hb^{\frac{1}{6}})
\end{align}
as $\hb\downarrow0$ uniformly for $\z\geq0$ and $\al\in(0,\al_0]$.

\begin{remark}
The  special case $\al=0$ (i.e. when equation 
(\ref{schrodi-liouville-final-form}) has a double turning point at $\z=0$)
satisfies the same estimates. Indeed, as in the proof of Theorem 
\ref{main-thm}, $|\ps(\z,0)|/\z^{\frac{1}{3}}$ is integrable at 
$\z=+\infty$. Furthermore $l(0)$ is independent of $\hb$ and from 
the formula $\Om(x)=1+|x|^{\frac{1}{3}}$ we see that the estimates 
above remain unchanged.
\end{remark}

\section{Connection Formulae}
\label{connection-formulas}

The results obtained so far are somewhat inadequate
because Theorem \ref{main-thm} defines the character of solutions of 
equation (\ref{schrodi-liouville-final-form}) only for non-negative values of 
$\z$. Indeed, we are incapable of constructing error bounds
like those ones in (\ref{junk1}) and (\ref{junk2}) for negative
$\z$, a drawback pertinent to the nature of parabolic cylinder functions
(cf. Miller's \cite{miller1950}).

Consider $Y_1$ for example
\footnote{
Similar thinking can be argued for $Y_2$ as well.
}
. 
As $\hb\downarrow0$ in a continuous manner, the 
asymptotic behavior of its approximant 
$U(\z\sqrt{2\hb^{-1}},-\tfrac{1}{2}\hb^{-1}\al^2)$ 
at $\z=-\infty$, changes abruptly as $\hb^{-1}\al^2$ goes through odd positive
integers (cf. \textit{exceptional values} in appendix 
\ref{parabolic-cylinder-functions}). 
$Y_1$ on the other hand is not expected
to exhibit the same change at exactly the same values of $\hb^{-1}\al^2$.

But we can determine the 
asymptotic behavior of $Y_1, Y_2$ for small $\hb>0$ and $\z<0$ by establishing 
appropriate \textit{connection formulae}.  Since 
$|\ps(\z,\al)|/|\z|^{\frac{1}{3}}$ is integrable at $\z=\pm\infty$ uniformly
with respect to $\al$, we can replace $\z$ by $-\z$ and appeal to Theorem
\ref{main-thm} to ensure two more solutions $Y_3, Y_4$ of equation 
(\ref{schrodi-liouville-final-form}) satisfying
\begin{align}
\label{y3-approx}
Y_3(\z,\al,\hb) & =
U(-\z\sqrt{2\hb^{-1}},-\tfrac{1}{2}\hb^{-1}\al^2)+
\frac{\msf(-\z\sqrt{2\hb^{-1}},-\tfrac{1}{2}\hb^{-1}\al^2)}
{\esf(-\z\sqrt{2\hb^{-1}},-\tfrac{1}{2}\hb^{-1}\al^2)}
\asympt(\hb^{\frac{2}{3}})\\
\nn
Y_4(\z,\al,\hb) & =
\ol U(-\z\sqrt{2\hb^{-1}},-\tfrac{1}{2}\hb^{-1}\al^2)+\\
\label{y4-approx}
&\hspace{2cm}\esf(-\z\sqrt{2\hb^{-1}},-\tfrac{1}{2}\hb^{-1}\al^2)
\msf(-\z\sqrt{2\hb^{-1}},-\tfrac{1}{2}\hb^{-1}\al^2)
\asympt(\hb^{\frac{2}{3}})
\end{align}
as $\hb\downarrow0$ uniformly for $\z\leq0$ and $\al\in[0,\al_0]$.
We express $Y_1, Y_2$ in terms of $Y_3, Y_4$ and write
\begin{align}\label{sigma1}
Y_1(\z,\al,\hb) & =
\s_1^1Y_3(\z,\al,\hb)+\s_1^2Y_4(\z,\al,\hb)\\
\label{sigma2}
Y_2(\z,\al,\hb) & =
\s_2^1Y_3(\z,\al,\hb)+\s_2^2Y_4(\z,\al,\hb).
\end{align}
The connection will become clear once  we  find approximations
for the coefficients $\s_i^j$, $i,j=1,2$ in the linear relations
(\ref{sigma1}) and (\ref{sigma2}).

Evaluating at $\z=0$ both equations (\ref{sigma1}) and (\ref{sigma2})
and their derivatives, after algebraic manipulations we obtain
\be\label{sigma}
\s_i^j=(-1)^{j+1}
\frac
{\W[Y_i(\cdot,\al,\hb),Y_{5-j}(\cdot,\al,\hb)](0)}
{\W[Y_3(\cdot,\al,\hb),Y_4(\cdot,\al,\hb)](0)}
\quad\text{for}\quad
i,j=1,2.
\ee
But using the results and properties of Parabolic Cylinder Functions and their
auxiliary functions from section \ref{parabolic-cylinder-functions} in the 
appendix, we find
\begin{align*}
Y_1(0,\al,\hb)&=
\msf(0)[\sin\vphi+\asympt(\hb^{\frac{2}{3}})]\\
Y_2(0,\al,\hb)&=
\msf(0)[\cos\vphi+\asympt(\hb^{\frac{2}{3}})]\\
Y_3(0,\al,\hb)&=
\msf(0)[\sin\vphi+\asympt(\hb^{\frac{2}{3}})]\\
Y_4(0,\al,\hb)&=
\msf(0)[\cos\vphi+\asympt(\hb^{\frac{2}{3}})]\\
\frac{\partial Y_1}{\partial\z}(0,\al,\hb)&=
-\sqrt{2\hb^{-1}}\nsf(0)[\cos\vphi+\asympt(\hb^{\frac{2}{3}})]\\
\frac{\partial Y_2}{\partial\z}(0,\al,\hb)&=
\sqrt{2\hb^{-1}}\nsf(0)[\sin\vphi+\asympt(\hb^{\frac{2}{3}})]\\
\frac{\partial Y_3}{\partial\z}(0,\al,\hb)&=
\sqrt{2\hb^{-1}}\nsf(0)[\cos\vphi+\asympt(\hb^{\frac{2}{3}})]\\
\frac{\partial Y_4}{\partial\z}(0,\al,\hb)&=
-\sqrt{2\hb^{-1}}\nsf(0)[\sin\vphi+\asympt(\hb^{\frac{2}{3}})]
\end{align*}
as $\hb\downarrow0$ where $\vphi=(1+\hb^{-1}\al^2)\frac{\pi}{4}$. Finally 
substituting these estimates in (\ref{sigma}) we obtain
\begin{align}
\label{sigma11}
\sigma_1^1&=\sin(\tfrac{1}{2}\pi\hb^{-1}\al^2)+\asympt(\hb^{\frac{2}{3}})\\
\label{sigma12}
\sigma_1^2&=\cos(\tfrac{1}{2}\pi\hb^{-1}\al^2)+\asympt(\hb^{\frac{2}{3}})\\
\nn
\sigma_2^1&=\cos(\tfrac{1}{2}\pi\hb^{-1}\al^2)+\asympt(\hb^{\frac{2}{3}})\\
\nn
\sigma_2^2&=-\sin(\tfrac{1}{2}\pi\hb^{-1}\al^2)+\asympt(\hb^{\frac{2}{3}})
\end{align}
as $\hb\downarrow0$ uniformly for $\al\in[0,\al_0]$.

\section{A Quantization Condition for Eigenvalues}
\label{quanta-evs}

In this section, we will derive information about the EVs of 
(\ref{ev-problem}) by assembling the results of the previous paragraphs. 
This process will be facilitated by the equivalent equation 
(\ref{schrodi-liouville-final-form}) where EVs appear for those values
of $\al$ for which there exists a solution that is decaying at both 
$\z=-\infty$ and $\z=+\infty$ of the real line. In the end, this approach 
will help us establish a quantization condition for the EVs of 
the Dirac operator $\mathfrak{D}_{\hbar}$ and their corresponding 
norming constants. We have the following theorem.
\begin{theorem}\label{BoSo}
Suppose that $\lam=i\mu\in\{i\kappa\mid\kappa\in[A_0,A_{max}]\}$ is an 
EV of the operator $\mathfrak{D}_{\hbar}$ (see (\ref{dirac})) and 
consider the $a>0$ such that $\m=A(a)$. There exists a non-negative 
integer $n$ (depending both on $\m$ and $\hb$) for which the 
Bohr-Sommerfeld quantization condition is satisfied, i.e.
\begin{equation}\label{BS-condition}
\displaystyle\int_{-a}^{a}
\big[A^{2}(x)-\mu^2\big]^{1/2}dx=\pi\Big(n+\tfrac{1}{2}\Big)\hbar+
\mathcal{O}(\hbar^{\frac{5}{3}})
\quad\text{as}\quad\hbar\downarrow0.
\end{equation}
Conversely, for every non-negative integer $n$ such that 
$\pi(n+\frac{1}{2})\hbar \in[0,\tfrac{\pi}{2}\al_0^2]$
there exists a unique eigenvalue $\lambda =i\mu$ 
of $\mathfrak{D}_{\hbar}$
and consequently an $a>0$ with
$\mu=A(a)$ (where both $\mu$ and $a$ depend on $n,\hb$) 
satisying
\begin{equation}\nn
\Bigg|
\displaystyle\int_{-a}^{a}
\big[A^{2}(x)-\mu^2\big]^{1/2}dx
-\pi\Big(n+\frac{1}{2}\Big)\hb\Bigg|\leq C\hbar^{\frac{5}{3}}
\end{equation}
with a constant $C$ depending neither on $n$ nor on $\hbar$.
\end{theorem}
\begin{proof}
For the first part of the theorem, we observe the following.
By referring to the asymptotic form of $Y_1(\z,\al,\hb)$ as $\z\to+\infty$
and the asymptotics for $Y_3(\z,\al,\hb)$ and $Y_4(\z,\al,\hb)$
as $\z\to-\infty$ (see (\ref{y1-approx}), (\ref{e1-asympt}), 
(\ref{y3-approx}) and (\ref{y4-approx})), equation (\ref{sigma1})
implies that in the presense of an EV, the coefficient $\s_1^2$ has to
be zero. Accordingly, by (\ref{sigma12}) we have
\be\nn
\cos(\tfrac{1}{2}\pi\hb^{-1}\al^2)=\asympt(\hb^{\frac{2}{3}})
\quad\text{as}\quad\hb\downarrow0
\ee
or equivalently, there is a non-negative integer such that 
\be\label{alpha-squared}
\al^2=(2n+1)\hb+\asympt(\hb^{\frac{5}{3}})
\quad\text{as}\quad\hb\downarrow0.
\ee
In view of (\ref{alpha}), this is exactly what we wanted.

To the converse now
\footnote{Here we follow Yafaev's idea found in \S 4 of \cite{yafa2011}.}.
Let us first deal with the existence. Define the map
$\fisf:[0,a_{0}]\rightarrow\R$ by
\be\label{fi-def}
\fisf(a):=\frac{\pi}{2}\al^2(a)=\int_{-a}^{a}[-f(t,a)]^{\frac{1}{2}}dt
\ee
(cf. (\ref{alpha}) and/or the LHS of (\ref{BS-condition})).
Fix a non-negative integer $n$ 
such that $\pi(n+1/2)\hbar$ belongs to a neighborhood of 
$\fisf(\ti a)$ where $A(\ti a)=\ti\m$ and $\al(\ti a)=\ti\al$.
From $(\ref{U-asympt})$ we know that the functions
$Y_1$, $Y_3$ belong to $L^{2}(\mathbb{R}_{+})$ and 
$L^{2}(\mathbb{R}_{-})$ respectively. Define the function
\be\label{definition-sigma}
\sigma(a,\hb):=\sigma_1^2(\al(a),\hb).
\ee
It is enough to show that $\sigma$ vanishes for some $a_{n}(\hb)$ 
satisfying
\be\nn
\Big|
\fisf\big(a_{n}(\hb)\big)
-\pi\Big(n+\frac{1}{2}\Big)\hbar\Big|\leq C\hbar^{\frac{5}{3}}.
\ee
and set
\be\nn
\al(a_{n}(\hb))=\al_{n}(\hb).
\ee
Using (\ref{alpha}) and Leibniz's rule we have
\begin{equation}\label{fi-prime}
\frac{d\fisf}{da}(a)=
-A(a)A'(a)
\int_{-a}^{a}\big[A^{2}(t)-A^{2}(a)\big]^{-1/2}dt
>0.
\end{equation}

This result tells us that $\fisf$ maps a neighborhood 
$(a_{1},a_{2})$ of $\ti a$ in a one-to-one way onto the neighborhood
$(\fisf(a_{1}),\fisf(a_{2}))$ of $\fisf(\ti a)$. Let 
$\xsf=\fisf(a)$, $a\in[0,a_{0}]$, $\ti\xsf=\fisf(\ti a)$ and set
\begin{equation*}
\chi(\xsf,\hbar):=
\sigma\big(\fisf^{-1}(\xsf),\hb\big)-\cos(\hb^{-1}\xsf)
,\quad\xsf\in\fisf\big([0,a_{0}]\big).
\end{equation*} 
By definition (\ref{definition-sigma}) of $\sigma$ and 
(\ref{sigma12}) we have 
\be\nn
|\chi(\xsf,\hb)|\leq C\hb^{\frac{2}{3}}
\ee 
for $\xsf$ in a neighborhood of $\ti\xsf$, where once more the
constant $C$ is independent of $\hbar$ and $\xsf$. 

With the above definitions, our equation now reads
\begin{align*}
0 & = \sigma(a,\hb)\\
  & = \chi(\xsf,\hb)+\cos(\hb^{-1}\xsf).
\end{align*}
So this equation has to have a solution $\xsf_{n}(\hbar)$
satisfying the estimate
\begin{equation*}
\Big|\xsf_{n}(\hbar)-\pi\Big(n+\frac{1}{2}\Big)\hbar\Big|\leq C\hbar^{\frac{5}{3}}.
\end{equation*} 
A change of variables $s=\hbar^{-1}\xsf$ transforms our 
problem to the equivalent assertion that equation
\begin{equation}\label{s-eq}
\chi(\hbar s,\hbar)+\cos s=0
\end{equation}
has to have a solution with respect to $s$, namely $s_{n}(\hbar)$,
such that
\begin{equation}\label{s-estim}
\Big|s_{n}(\hbar)-\pi\Big(n+\frac{1}{2}\Big)\Big|\leq C\hbar^{\frac{2}{3}}.
\end{equation}
But this is true because 
\begin{equation*}
\chi(\hbar s,\hbar)=
\mathcal{O}(\hbar^{\frac{2}{3}})
\quad\text{as}\quad\hbar\downarrow0.
\end{equation*}

To complete the proof of the theorem, we need uniqueness 
as well. Once again fix $n\in\mathbb{Z}$. We have just proved that
for this $n$, equation (\ref{s-eq}) has a solution obeying (\ref{s-estim}).
We shall employ reductio ad absurdum. Suppose, on the contrary,
that there are $s_{1}, s_{2}$ - with $s_{1}<s_{2}$ - satisfying
(\ref{s-estim}) so that the function
\begin{equation*}
g(s):=\chi(\hbar s,\hbar)+\cos s
,\quad s\in[s_{1},s_{2}]
\end{equation*}
is zero; $g(s_{1})=g(s_{2})=0$. Furthermore, $g$ is continuous in 
$[s_{1},s_{2}]$ and differentiable in $(s_{1},s_{2})$ with
\begin{equation*}
g'(s)=\hbar\frac{\partial\chi}{\partial\xsf}
(\hbar s,\hbar)-\sin s,\quad s\in(s_{1},s_{2}).
\end{equation*}
By Rolle's theorem there is $\tilde{s}\in(s_{1},s_{2})$
such that
\begin{align*}
0 & = g'(\tilde{s})\\
  & = \hbar\frac{\partial\chi}{\partial\xsf}
(\hbar\tilde{s},\hbar)-\sin\tilde{s}.
\end{align*}

Recapping, we have found
\begin{itemize}
\item
$\tilde{s}\in(s_{1},s_{2})$ which says that
$\tilde{s}$ satisfies (\ref{s-estim}) too; namely 
\begin{equation}\label{s-tilde}
\tilde{s}=\pi\Big(n+\frac{1}{2}\Big)+\mathcal{O}(\hbar^{\frac{2}{3}})
\quad\text{as}\quad\hb\downarrow0
\end{equation}
\item
$\tilde{s}$ is a root of the equation
\begin{equation}\label{false}
\sin s=\hbar\frac{\partial\chi}{\partial\xsf}
(\hbar\tilde{s},\hbar).
\end{equation}
\end{itemize} 
Using (\ref{s-tilde}), the left-hand side of (\ref{false}) is seen
to be $(-1)^{n}$ as $\hbar\downarrow0$. Now, using (\ref{sigma12})  
observe that
\begin{equation*}
\frac{\partial\sigma}{\partial a}(a,\hbar)=
-\hbar^{-1}\fisf'(a)\sin
\Big[\hbar^{-1}\fisf(a)\Big]+
\mathcal{O}(1)
\quad\text{as}\quad\hbar\downarrow0
\end{equation*}
which eventually leads to
\begin{equation*}
\frac{\partial\chi}{\partial\xsf}(\xsf,\hbar)=\mathcal{O}(1)
\quad\text{as}\quad\hbar\downarrow0.
\end{equation*}
Hence the right-hand side of (\ref{false})
is $\mathcal{O}(\hbar)$ as $\hbar\downarrow0$ which is
a contradiction. Thus, there is only one such eigenvalue.
\end{proof}

\begin{remark}
A result like equation (\ref{BS-condition}) can also be found in 
\cite{yafa2011} for the Schr\"odinger operator, with the slightly 
better asymptotic estimate of order $\hb^2$. Although the result 
we provide here is only $\asympt(\hb^\frac{5}{3})$, it has the 
additional advantage of holding for the critical case of a double 
turning point as well.
\end{remark}

The following corollary is a straightforward application of the  
Theorem \ref{BoSo} giving the number of EVs
of the Dirac operator $\mathfrak{D}_{\hbar}$ in a fixed
(independent of $\hb$) interval not containing 0, on the imaginary axis.
\begin{corollary}\label{ev-count-yafa}
Consider an interval $(\m_1,\m_2)\subset[A_{0},A_{max}]$ and take 
$a_j$, $j=1,2$ such that $A(a_j)=\m_j$ for $j=1,2$. Then the total number
$\mathcal{N}_{\hbar}$ of eigenvalues $\lambda=i\mu$ of the Dirac operator 
$\mathfrak{D}_{\hbar}$ lying in the set 
$\{i\mu\mid \mu\in(\m_1,\m_2)\}\subset\mathbb{C}$
is equal to
\begin{equation}\label{EV-count}
\mathcal{N}_{\hbar}=\pi^{-1}\big[\fisf(a_1)-\fisf(a_2)\big]
\hbar^{-1}+R(\hbar)
\end{equation}
where $|R(\hbar)|\leq1$ for sufficiently small $\hbar$. 
\end{corollary}

\begin{figure}[H]
\centering
\begin{tikzpicture}

\draw (6.5,3) node {$\m$-space};
\draw (4,2) -- (9.5,2);
\draw [(-), thick, black] (5,2) -- (8,2)
node[anchor=south west] {$\m_2$};
\draw [(-), thick, black] (8,2) -- (5,2)
node[anchor=north east] {$\m_1$};

\node[circle,inner sep=2pt,fill=black,label=below:
{$\m$}] at (6,2) {};

\draw (0,0) -- (6,0);
\draw [(-), ultra thick, olive] (0.5,0) -- (5.5,0)
node[anchor=south west] {$a_1$};
\draw [(-), ultra thick, olive] (5.5,0) -- (0.5,0)
node[anchor=north east] {$a_2$};

\node[circle,inner sep=2pt,fill=black,label=below:
{$a$}] at (4.5,0) {};

\fill[fill=blue, opacity=0.2](4,0.2)--(4.75,0.2)--(4.75,-0.2)--
(4,-0.2)--(4,0.2);

\draw (2,1) node {$a$-space};

\draw (10,1) node {$\xsf$-space};

\draw (7,0) -- (13,0);
\draw [(-), ultra thick, red] (8,0) -- (12,0)
node[anchor=south west] {$\fisf(a_1)$};
\draw [(-), ultra thick, red] (12,0) -- (8,0)
node[anchor=north east] {$\fisf(a_2)$};

\node[circle,inner sep=2pt,fill=black,label=below:
{$\pi(n+\tfrac{1}{2})\hbar$}] at (11,0) {};

\draw[thick, ->] (3.65,-1) arc (235:305:5);

\draw (8.5,-2) node {$\fisf$};
\end{tikzpicture}
\caption{Counting eigenvalues using the Bohr-Sommerfeld condition.}
\end{figure}
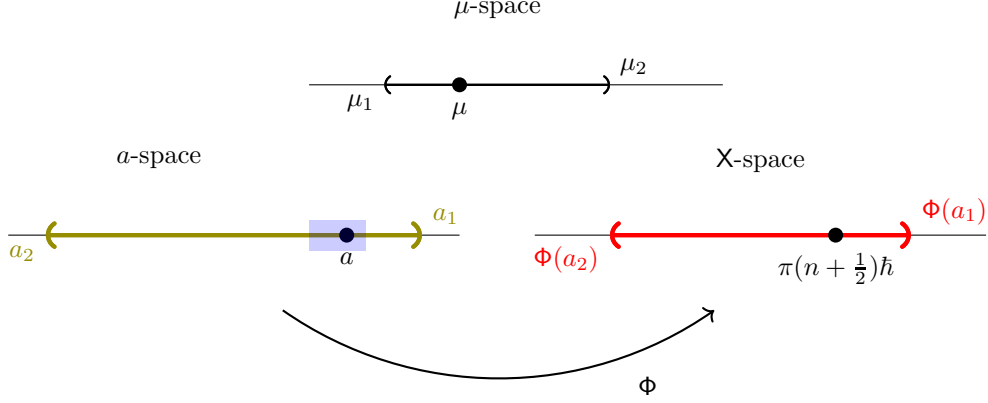

\begin{proof}
Observe that
\be\nn
\mathcal{N}_{\hbar} =
\#\big\{a\in(a_1,a_2)\mid
\text{equation (\ref{schrodi_final-version}) with}\hspace{.15cm}
\m=A(a)\hspace{.15cm}
\text{has an}\hspace{.15cm}
L^2(\R)\hspace{.15cm}\text{solution}\big\}
\ee
By Theorem \ref{BoSo}, there is only one \say{$a$-eigenvalue}  
in a neighborhood of length 
$C\hbar^{\frac{5}{3}}$ of every point $\fisf^{-1}\big(\pi(n+1/2)\hbar\big)$. 
For sufficiently small
$\hbar$ these neighborhoods are mutually disjoint. But this means that
the number $\mathcal{N}_{\hbar}$ is equal to the number of the points
$\pi(n+1/2)\hbar$ that lie in the interval 
$\big(\fisf(a_2),\fisf(a_1)\big)$, i.e.
\be\nn
\mathcal{N}_{\hbar} =
\#\big\{n\in\Z\big| 
\pi(n+1/2)\hbar\in\big(\fisf(a_2),\fisf(a_1)\big)\big\}
\ee
for sufficiently small $\hbar$. And this number is exactly 
$\pi^{-1}\big[\fisf(a_1)-\fisf(a_2)\big]
\hbar^{-1}+R(\hbar)$ with $|R(\hbar)|\leq1$.
\end{proof}

\begin{remark}
(Weyl's formula)
Using the definition (\ref{fi-def}), we can write $\fisf$ in a different way.
Indeed, we have
\begin{align*}
\fisf(a) & =
\int_{-a}^{a}\Big[A^{2}(x)-\mu^{2}\Big]^{1/2}dx=
\frac{1}{2}\int_{-a}^{a}2\Big[A^{2}(x)-\mu^{2}\Big]^{1/2}dx=\\
& = \frac{1}{2}\iint_{A^{2}(x)-k^{2}\geq\m^{2}}dkdx
\end{align*}
With the help of this last equality, the difference
$\fisf(a_1)-\fisf(a_2)$ in (\ref{EV-count}) can be equivalently written as:
\begin{align*}
\fisf(a_1)-\fisf(a_2) & =
\frac{1}{2}\iint_{A^{2}(x)-k^{2}\geq\m_{1}^2}dkdx-
\frac{1}{2}\iint_{A^{2}(x)-k^{2}\geq\m_{2}^2}dkdx\\
& = \frac{1}{2}
    \iint_{\m_1^{2}\leq A^{2}(x)-k^{2}\leq\m_2^{2}}dkdx\\
& = \frac{1}{2}\cdot
    Area\Big(\{(x,k)\in\mathbb{R}^{2}\big|\m_1^{2}\leq A^{2}(x)-k^{2}
    \leq\m_2^{2}\}\Big)    
\end{align*}
which means that the asymptotic coefficient in (\ref{EV-count}) is
the area of a region in the phase space.
Consequently, relation (\ref{EV-count}) is 
the WKB analogue of Weyl's formula with a strong estimate on the remainder.
\end{remark}

Another straightforward application of Theorem \ref{BoSo} allows us
to express the norming constants of the Dirac operator 
$\mathfrak{D}_{\hbar}$. In particular we see that the asymptotics  
obtained agree with the known fact that (because of the symmetry of 
the potential, see Chapter 3 of \cite{kmm}) the corresponding norming 
constant is exactly $(-1)^n$ for some integer $n$. But of course, our 
method is easily extensible to the non-symmetric case, which we will 
in fact consider in a sequel to this paper. We have the following 
corollary.
\begin{corollary}
\label{norm-const}
Suppose that $\lam(\hb)$ is an EV of $\mathfrak{D}_{\hbar}$.
Then there is a non-negative integer $n$ (depending both on $\lam$ 
and $\hb$) such that the corresponding norming constant has 
asymptotics 
\be\nn
(-1)^n + \asympt(\hb^\frac{2}{3})
\quad\text{as}\quad
\hbar\downarrow0.
\ee
\end{corollary}
\begin{proof}
If $\lam(\hb)=i\mu(\hb)$ then by Theorem \ref{BoSo}, (\ref{alpha-squared}) 
and (\ref{sigma11}) there is a non-negative integer $n$ such that
\be\nn
\sigma_1^1=(-1)^n+\asympt(\hb^\frac{2}{3})
\quad\text{as}\quad\hb\downarrow0
\ee
where $n$ is the same as in (\ref{alpha-squared}). Thus for the 
corresponding \say{$\al$-eigenvalue}, namely $\al(\hb)$, we have
\be\nn
Y_1(\z,\al(\hb),\hb) = \big[(-1)^n+\asympt(\hb^\frac{2}{3})\big]
Y_3(\z,\al(\hb),\hb)
\quad\text{as}\quad
\hb\downarrow0
\ee
and this proves the assertion.
\end{proof} 

\section{Eigenvalues Near Zero}
\label{near-zero-evs}

It is important for the applications to the semiclassical theory of 
the focusing NLS equation, to understand the behavior of the EVs near 
0. As done in  \cite{fujii+kamvi} we will examine here two specific 
-but quite inclusive- families of data $A$: the asymptotically rational 
case and the asymptotically exponential case. We do note however that 
the formula (\ref{psi-h}) we have for the function $\psi$, makes it very 
easy to check if the behavior of the EVs near zero is good enough
for any family of $A$, defined by explicit prescribed asymptotics at 
infinity
\footnote{
Indeed, apart from the  the asymptotically rational case and the 
asymptotically exponential case presented here, we have also done so 
for functions $A$ with asymptotics $x^n\exp(-|x|^d)$ and 
$(\log|x|)^n |x|^{-d-1} $
where $n\geq0$ is an integer and $d\in\R_+$.
}
.
 
To begin with, we choose $b>0$ independent of $\hb$ and set 
$\m(\hb)=\hb^b$. Consider the equation $\m(\hb)=A(a(\hb))$. 
Observe that $\m(\hb)\downarrow0$ as $\hb\downarrow0$ while 
$a(\hb)\uparrow+\infty$ as $\hb\downarrow0$. Furthermore,
(recalling the notation of section \S\ref{from-schrodi-to-dirac}) 
$0<a_{0}(\hb)=B(\hb^b)$ where $B$ is the inverse of 
$A\big|_{[0,+\infty)}$. Consequently, if we use the abbreviation
\be\nn
\al(\hb)\equiv\al(a(\hb))
\ee 
the above can be translated in the $\al$-space as 
$0<\al_{0}(\hb)\equiv\al(a_0(\hb))<+\infty$; observe that 
$B(\hb^b)\uparrow+\infty$ as $\hb\downarrow0$ and hence by 
the definition of $\al$ (see (\ref{alpha})) we obtain
$$\al(\hb)\uparrow\big(\tfrac{2}{\pi}\|A\|_{L^1(\R)}\big)^{\frac{1}{2}}
\quad\text{as}\quad\hb\downarrow0.$$ 

In this setting, $a$ depends on $\hb$ and using (\ref{v-new}), 
(\ref{f-new}) and (\ref{p}) we have
\be\label{f-h}
f\big(x,a(\hb)\big)=A^2\big(a(\hb)\big)-A^2(x)
\ee
\be\label{g-h}
g\big(x,a(\hb)\big)=
\frac{3}{4}\Big[\frac{A'(x)}{A(x)+A\big(a(\hb)\big)}\Big]^2-
\frac{1}{2}\frac{A''(x)}{A(x)+A\big(a(\hb)\big)}
\ee 
and
\be\nn
f(x,a(\hb))=[x^2-a^2(\hb)]p(x,a(\hb)).
\ee
It is easy to see that for each value of $\hb$ the functions $f$, 
$g$ and $p$ satisfy properties $(i)$ through $(iv)$ of 
Lemma \ref{junk-continuity} in \S\ref{continuity}. This implies 
-again with the help of Lemma I in \cite{olver1975}- that for each 
$\hb$ the function
\begin{multline}\label{psi-h}
\ps(\z,\al(\hb))=
\frac{1}{4}\frac{3\z^2 +2\al^2(\hb)}{[\z^2 -\al^2(\hb)]^2}
+\frac{1}{16}\frac{\z^2 -\al^2(\hb)}{f^3(x,a(\hb))}\\
\cdot\Big\{4f(x,a(\hb))f''(x,a(\hb))-5[f'(x,a(\hb))]^2\Big\}+
[\z^2 -\al^2(\hb)]\frac{g(x,a(\hb))}{f(x,a(\hb))}
\end{multline}
is continuous in the corresponding region of the $(\z,\al)$-plane. 

So in order to have a conclusion such as Theorem \ref{main-thm} 
-and eventually a result like Theorem \ref{BoSo}- we need to investigate 
the convergence of the integral in (\ref{variation-total}) but now
with an $\al$ that depends on $\hb$, i.e.
\be
\label{hbar-total-var}
\int_{0}^{+\infty}\frac{|\ps(t,\al(\hb))|}{\Om(t\sqrt{2\hb^{-1}})}dt.
\ee

Before considering  our two specific families of data, we would like to 
remind the reader of Lemma \ref{x-at-big-zeta} and especially the 
formula (\ref{x-asymptotics}); an asymptotic relation which now reads
\be\label{x-vs-z}
x=\frac{\z^2}{2A(a(\hb))}\Big[1+\asympt\Big(\tfrac{\log\z}{\z^2}\Big)\Big]
\quad
\text{as}
\hspace{7pt}
\z\uparrow+\infty.
\ee
It shall be used to allow us understand the nature of $\ps$ for big $\z$. 

\subsection{The Rational Case}

For the moment, assume for simplicity that  in addition to the 
assumptions of paragragh \S\ref{assume}, $A$ also  satisfies
\be\label{a-rational}
A(x)=\frac{1}{|x|^{d}}\quad\text{for}\quad|x|\geq1
\ee
where $d>1$.  In this case, using (\ref{x-vs-z}) we get 
\be
\label{x-vs-z-rational}
x=\frac{1}{2}a^d(\hb)\z^2\Big[1+\asympt\Big(\tfrac{\log\z}{\z^2}\Big)\Big]
\quad
\text{as}
\hspace{7pt}
\z\uparrow+\infty.
\ee

Using (\ref{f-h}), (\ref{g-h}), (\ref{psi-h}), (\ref{a-rational}) and 
(\ref{x-vs-z-rational}) we arrive at
\be
\ps(\z,\al(\hb))=
\ps_{1}(\z,\al(\hb))\Big[1+\asympt\Big(\tfrac{\log\z}{\z^2}\Big)\Big]
\quad
\text{as}
\hspace{7pt}
\z\uparrow+\infty
\ee
uniformly in $\al$ and consequently in $\hb$, where
\begin{align*}
\ps_{1}(\z,\al(\hb))&=
\frac{1}{4}\frac{3\z^2 +2\al^2(\hb)}{[\z^2 -\al^2(\hb)]^2}\\
& -
\frac{d}{4^d}a^{2d^2+2d}(\hb)\z^{4d-4}[\z^2 -\al^2(\hb)]
\frac{\frac{2(2d+1)}{4^d}a^{2d^2}(\hb)\z^{4d}+(d-2)a^{2d}(\hb)}
{[\frac{1}{4^d}a^{2d^2}(\hb)\z^{4d}-a^{2d}(\hb)]^3}\\
& +
\frac{d}{4^d}a^{2d^2+d}(\hb)\z^{4d-4}[\z^2 -\al^2(\hb)]
\frac{-\frac{d+1}{2^{d-1}}a^{d^2}(\hb)\z^{2d}+(d-2)a^{d}(\hb)}
{[\frac{a^{d^2}(\hb)}{2^d}\z^{2d}-a^{d}(\hb)]
[\frac{a^{d^2}(\hb)}{2^d}\z^{2d}+a^{d}(\hb)]^3}.
\end{align*}

Similar considerations to the ones just presented can be applied to a more
general $A$ and the result still remains the same.  We state the following assumption.
\begin{assumption}
\label{a-rational-gen}
Consider a function $A$ satisfying Assumption \ref{primary-assume} and such that
\be
A(x)=C\frac{1}{|x|^{d}}+r(x)\quad\text{for}\quad|x|\geq1
\ee
where $C$ is a positive constant, $d>1$ and $r$ is a function satisfying
Assumption \ref{primary-assume} along with the asymptotics
\begin{itemize}
\item
$r(x)=o(|x|^{-d})$ as $x\to\pm\infty$
\item
$r'(x)=o(|x|^{-d-1})$ as $x\to\pm\infty$
\item
$r''(x)=o(|x|^{-d-2})$ as $x\to\pm\infty$.
\end{itemize}
\end{assumption}
Observe that the potential in the example treated in 
paragraph \S\ref{ex} belongs to this case.  Hence from now on in this 
subsection, when we write $A$ we mean a special one from  satisfying
this Assumption \ref{a-rational-gen}.

The asymptotics above  imply that for each $\hb$, the integral in 
(\ref{hbar-total-var}) converges; furthermore, this convergence is 
uniform in $a(\hb)$. This means that a variation of Theorem 
\ref{main-thm} can be applied to guarantee the existence of
approximate functions in this case too. Indeed, Theorem 
\ref{thm-on-exist-int-eq} comes into play and guarantees that 
everything remains unchanged; for each value of $\hb$, equation 
(\ref{schrodi-liouville-final-form}), i.e.
\be
\label{h-dependent-equation}
\frac{d^2Y}{d\z^2}=
\big[\hb^{-2}\big(\z^2-\alpha^2(\hb)\big)+\ps(\z,\al(\hb))\big]Y
\ee
has in the region $[0,+\infty)\times[0,\al_{0}(\hb)]$ of the 
$(\z,\al)$-plane solutions $Y_+$ and $\ol Y_+$ 
(being extensions in $\al$ of $Y_1$ and $Y_2$ respectively, 
cf. (\ref{y1-approx}), (\ref{y2-approx})) which are continuous, 
have continuous first and second partial $\z$-derivatives, and are 
given by 
\bea\nn
Y_+(\z,\al(\hb),\hb)=
U(\z\sqrt{2\hb^{-1}},-\tfrac{1}{2}\hb^{-1}\al^2(\hb))+
\eps(\z,\al(\hb),\hb)\\
\nn
\ol{Y}_+(\z,\al(\hb),\hb)=
\ol U(\z\sqrt{2\hb^{-1}},-\tfrac{1}{2}\hb^{-1}\al^2(\hb))+
\ol{\eps}(\z,\al(\hb),\hb)
\eea
(cf. (\ref{y1-approx}), (\ref{y2-approx})) where for the remainders
we have the relations
\begin{multline}\nn
\frac{|\eps(\z,\al(\hb),\hb)|}
{\msf(\z\sqrt{2\hb^{-1}},-\tfrac{1}{2}\hb^{-1}\al^2(\hb))},
\frac{\Big|\frac{\partial \eps}
{\partial\z}(\z,\al(\hb),\hb)\Big|}{\sqrt{2\hb^{-1}}
\nsf(\z\sqrt{2\hb^{-1}},-\tfrac{1}{2}\hb^{-1}\al^2(\hb))}\\
\leq
\frac{1}{\esf(\z\sqrt{2\hb^{-1}},-\tfrac{1}{2}\hb^{-1}\al^2(\hb))}
\Big(
\exp\big\{
\tfrac{1}{2}(\pi\hb)^{\frac{1}{2}}l(-\tfrac{1}{2}\hb^{-1}\al^2(\hb))
\mathcal{V}_{\z,+\infty}[H](\al(\hb),\hb)
\big\}
-1
\Big)
\end{multline}
and
\begin{multline}\nn
\frac{|\ol{\eps}(\z,\al(\hb),\hb)|}
{\msf(\z\sqrt{2\hb^{-1}},-\tfrac{1}{2}\hb^{-1}\al^2(\hb))},
\frac{\Big|\frac{\partial\ol{\eps}}{\partial\z}(\z,\al(\hb),\hb)\Big|}
{\sqrt{2\hb^{-1}}
\nsf(\z\sqrt{2\hb^{-1}},-\tfrac{1}{2}\hb^{-1}\al^2(\hb))}\\
\leq
\esf(\z\sqrt{2\hb^{-1}},-\tfrac{1}{2}\hb^{-1}\al^2(\hb))
\Big(
\exp\big\{
\tfrac{1}{2}(\pi\hb)^{\frac{1}{2}}l(-\tfrac{1}{2}\hb^{-1}\al^2(\hb))
\mathcal{V}_{0,\z}[H](\al(\hb),\hb)
\big\}
-1
\Big)
\end{multline}
(analogous to (\ref{junk1}), (\ref{junk2})).

Additionally, $l$ and $\var_{0,+\infty}[H]$ satisfy the same 
asymptotics as before (cf. (\ref{lambda-asympt}), 
(\ref{v-operator-asympt})) and consequently one obtains the same 
asymptotic behavior of solutions as in \S\ref{asympt-behave-sols}; 
namely
\begin{align*}
\eps(\z,\al(\hb),\hb) & =
\frac{\msf(\z\sqrt{2\hb^{-1}},-\tfrac{1}{2}\hb^{-1}\al^2(\hb))}
{\esf(\z\sqrt{2\hb^{-1}},-\tfrac{1}{2}\hb^{-1}\al^2(\hb))}
\asympt(\hb^{\frac{2}{3}})
\\
\nn
\ol{\eps}(\z,\al(\hb),\hb) & =
\esf(\z\sqrt{2\hb^{-1}},-\tfrac{1}{2}\hb^{-1}\al^2(\hb))
\msf(\z\sqrt{2\hb^{-1}},-\tfrac{1}{2}\hb^{-1}\al^2(\hb))
\asympt(\hb^{\frac{2}{3}})
\\
\nn
\frac{\partial\eps}{\partial\z}(\z,\al(\hb),\hb) & =
\frac{\nsf(\z\sqrt{2\hb^{-1}},-\tfrac{1}{2}\hb^{-1}\al^2(\hb))}
{\esf(\z\sqrt{2\hb^{-1}},-\tfrac{1}{2}\hb^{-1}\al^2(\hb))}
\asympt(\hb^{\frac{1}{6}})
\\
\nn
\frac{\partial\ol{\eps}}{\partial\z}(\z,\al(\hb),\hb) & =
\esf(\z\sqrt{2\hb^{-1}},-\tfrac{1}{2}\hb^{-1}\al^2(\hb))
\nsf(\z\sqrt{2\hb^{-1}},-\tfrac{1}{2}\hb^{-1}\al^2(\hb))
\asympt(\hb^{\frac{1}{6}})
\end{align*}
as $\hb\downarrow0$ uniformly for $\z\geq0$ and 
$\al(\hb)\in[0,\al_0(\hb)]$.

Arguing as in \S\ref{connection-formulas}, we obtain two more 
solutions of (\ref{h-dependent-equation}), namely $Y_-$ and 
$\ol{Y}_-$ (the equivalent of $Y_3$ and $Y_4$ correspondingly), 
satisfying
\begin{align*}
Y_-(\z,\al(\hb),\hb) & =
U(-\z\sqrt{2\hb^{-1}},-\tfrac{1}{2}\hb^{-1}\al^2(\hb))+
\frac{\msf(-\z\sqrt{2\hb^{-1}},-\tfrac{1}{2}\hb^{-1}\al^2(\hb))}
{\esf(-\z\sqrt{2\hb^{-1}},-\tfrac{1}{2}\hb^{-1}\al^2(\hb))}
\asympt(\hb^{\frac{2}{3}})\\
\nn
\ol{Y}_-(\z,\al(\hb),\hb) & =
\ol U(-\z\sqrt{2\hb^{-1}},-\tfrac{1}{2}\hb^{-1}\al^2(\hb))+\\
&\hspace{2cm}\esf(-\z\sqrt{2\hb^{-1}},-\tfrac{1}{2}\hb^{-1}\al^2(\hb))
\msf(-\z\sqrt{2\hb^{-1}},-\tfrac{1}{2}\hb^{-1}\al^2(\hb))
\asympt(\hb^{\frac{2}{3}})
\end{align*}
as $\hb\downarrow0$ uniformly for $\z\leq0$ and 
$\al(\hb)\in[0,\al_0(\hb)]$.

Consequently we have the same connection formulae (all the results 
of \S\ref{connection-formulas} are not altered at all). Indeed, 
expressing $Y_+$, $\ol{Y}_+$ in terms of $Y_-$, $\ol{Y}_-$ and 
writing
\begin{align*}
Y_+(\z,\al(\hb),\hb) & =
\ta_1^1Y_-(\z,\al(\hb),\hb)+\ta_1^2\ol{Y}_-(\z,\al(\hb),\hb)\\
\ol{Y}_+(\z,\al(\hb),\hb) & =
\ta_2^1Y_-(\z,\al(\hb),\hb)+\ta_2^2\ol{Y}_-(\z,\al(\hb),\hb)
\end{align*}
(confer (\ref{sigma1}), (\ref{sigma2})) in the same way we find that
\begin{align*}
\ta_1^1&=\sin(\tfrac{1}{2}\pi\hb^{-1}\al^2(\hb))+
\asympt(\hb^{\frac{2}{3}})\\
\ta_1^2&=\cos(\tfrac{1}{2}\pi\hb^{-1}\al^2(\hb))+
\asympt(\hb^{\frac{2}{3}})\\
\ta_2^1&=\cos(\tfrac{1}{2}\pi\hb^{-1}\al^2(\hb))+
\asympt(\hb^{\frac{2}{3}})\\
\ta_2^2&=-\sin(\tfrac{1}{2}\pi\hb^{-1}\al^2(\hb))+
\asympt(\hb^{\frac{2}{3}})
\end{align*}
(like (\ref{sigma11}), (\ref{sigma12})) as $\hb\downarrow0$ uniformly 
for $\al(\hb)\in[0,\al_0(\hb)]$.

Eventually, this means that the results of \S\ref{quanta-evs} for 
the EVs remain the same (eg. Theorem \ref{BoSo} but this time for 
$\m(\hb)\in[\hb^b,A_{max}]$).
Hence, we arrive at the following theorem the proof of which has 
been already provided in the previous paragraph (cf. Theorem \ref{BoSo}).
\begin{theorem}
\label{boso-near-0}
Suppose that $\lam(\hb)=i\mu(\hb)$, where $\mu(\hb)\in[\hb^b,A_{max}]$ 
for an arbitrary $\hb$-independent positive constant $b$,
is an EV of the operator $\mathfrak{D}_{\hbar}$ (see (\ref{dirac})) with 
a potential $A$ satisfying Assumption \ref{a-rational-gen}. 
Consider $a(\hb)$ such that $\m(\hb)=A[a(\hb)]$. Then there exists a 
non-negative integer $n$ for which
\begin{equation}
\displaystyle\int_{-a(\hb)}^{a(\hb)}
\big[A^{2}(x)-\mu^2(\hb)\big]^{1/2}dx=\pi\Big(n+\tfrac{1}{2}\Big)\hbar+
\mathcal{O}(\hbar^{\frac{5}{3}})
\quad\text{as}\quad\hbar\downarrow0.
\end{equation}
Conversely, for every non-negative integer $n$ such that 
$\pi(n+\frac{1}{2})\hbar\in$
$[0,\tfrac{\pi}{2}\al_0^2(\hb)]$
there exists a unique eigenvalue $i\mu_{n}(\hbar)=iA[a_{n}(\hbar)]$ satisying
\begin{equation}\nn
\Bigg|
\displaystyle\int_{-a_{n}(\hbar)}^{a_{n}(\hbar)}
\big[A^{2}(x)-\mu_{n}(\hbar)^2\big]^{1/2}dx
-\pi\Big(n+\frac{1}{2}\Big)\hb\Bigg|\leq C\hbar^{\frac{5}{3}}
\end{equation}
with a constant $C$ depending neither on $n$ nor on $\hbar$.
\end{theorem}

Let us here state a useful definition.
\begin{definition}
We define $a_n^{WKB}(\hb)$  such that
\be\label{wkb-a-ev}
\displaystyle\int_{-a_n^{WKB}(\hb)}^{a_n^{WKB}(\hb)}
\big[A^{2}(x)-A^2(a_n^{WKB}(\hb))\big]^{1/2}dx=
\pi\Big(n+\tfrac{1}{2}\Big)\hbar
\ee
and set
\be\label{wkb-mu-ev}
\m_{n}^{WKB}(\hb):=A(a_n^{WKB}(\hb))
\ee
to be the \textbf{WKB-approximant} to an actual 
\say{$\m$-eigenvalue} $\m_n(\hb)$ satisfying 
$$\m_{n}(\hb)=A(a_n(\hb)).$$
\end{definition}
With this in mind,  we have the following corollary.
\begin{corollary}
Let $b>0$ (independent of $\hb$) and consider a function $A$
satisfying Assumption \ref{a-rational-gen}.
Then for every non-negative integer $n$ such that 
$\pi(n+\frac{1}{2})\hbar$ belongs to $(0,\tfrac{\pi}{2}\al_0^2(\hb))$,
there exists a unique eigenvalue $i\m_{n}(\hbar)$ satisfying
\begin{equation}\nn
|\m_n(\hb)-\m_n^{WKB}(\hb)|=
\asympt(\hbar^{\frac{5}{3}+\frac{b}{d}})
\quad\text{as}\quad\hb\downarrow0
\end{equation}
uniformly for $\mu_n(\hb)$ in $[\hb^b,A_{max}]$.
\end{corollary}
\begin{proof}
Using the result of the previous Theorem, (\ref{a-rational-gen}) and 
(\ref{fi-prime}) we have
\begin{align*}
|\m_n(\hb)-\m_n^{WKB}(\hb)|
& =
|A(a_n(\hb))-A(a_n^{WKB}(\hb))|\\
& \leq
C|A'(a)|[\fisf'(a)]^{-1}\hb^{\frac{5}{3}}
\quad(\text{for some $a$ between}\hspace{2pt}a_n(\hb),a_n^{WKB}(\hb))\\
& \leq
C a^{-d-1}a^d\hb^{\frac{5}{3}}\\
& \leq
C\hbar^{\frac{5}{3}+\frac{b}{d}}
\quad(\text{since}\hspace{4pt} a^{-1}\sim\hb^{\frac{b}{d}})
\end{align*}
where as usual $C$ denotes a generic constant.
\end{proof} 

\subsection{The Exponential Case}

In this subsection we  start with a function $A$ that satisfies the 
assumptions of \S\ref{assume} and furthermore
\be\label{a-expo}
A(x)=e^{-|x|^{\delta}},\quad\text{for}\quad|x|\geq1
\ee
where $\delta>0$. Now using (\ref{x-vs-z}) we get 
\be
\label{x-vs-z-exponentional}
x=\frac{1}{2}\z^2\exp(a^\delta(\hb))
\Big[1+\asympt\Big(\tfrac{\log\z}{\z^2}\Big)\Big]
\quad
\text{as}
\hspace{7pt}
\z\uparrow+\infty.
\ee 
To  simplify  notation, we set
\be\label{simply}
\Lambda(\hb)\equiv\Lambda(a(\hb)):=\frac{1}{2}\exp(a^\delta(\hb))
\ee

Using (\ref{f-h}), (\ref{g-h}), (\ref{psi-h}), (\ref{a-expo}),
(\ref{x-vs-z-exponentional}) and (\ref{simply}) we obtain
\be
\ps(\z,\al(\hb))=
\ps_{2}(\z,\al(\hb))\Big[1+\asympt\Big(\tfrac{\log\z}{\z^2}\Big)\Big]
\quad
\text{as}
\hspace{7pt}
\z\uparrow+\infty
\ee
uniformly in $\al$ and consequently in $\hb$, where
\begin{align*}
\ps_2(\z,\al(\hb))
& =
\frac{1}{4}\frac{3\z^2 +2\al^2(\hb)}{[\z^2 -\al^2(\hb)]^2}\\
& +
\frac{\de}{4}\z^{2\de-4}[\z^2 -\al^2(\hb)]
\frac{\La^{\de-2}(\hb)\exp\{-2\La^\de(\hb)\z^{2\de}\}}
{[\La^{-2}(\hb)-\exp\{-2\La^{\de}(\hb)\z^{2\de}\}]^3}\cdot\\
&
\hspace{1cm}
[
-\de\La^{\de}(\hb)\z^{2\de}\exp\{-2\La^{\de}(\hb)\z^{2\de}\}
-2(\de-1)\exp\{-2\La^{\de}(\hb)\z^{2\de}\}\\
&
\hspace{5.5cm}
-4\de\La^{\de-2}(\hb)\z^{2\de}
+2(\de-1)\La^{-2}(\hb)
]\\
& +
\frac{\de}{4}\z^{2\de-4}[\z^2 -\al^2(\hb)]\cdot\\
&
\hspace{1cm} 
\frac{\La^{\de-2}(\hb)\exp\{-\La^\de(\hb)\z^{2\de}\}}
{[\La^{-1}(\hb)-\exp\{-\La^{\de}(\hb)\z^{2\de}\}]
[\La^{-1}(\hb)+\exp\{-\La^{\de}(\hb)\z^{2\de}\}]^3}\cdot\\
&
\hspace{2cm}
[
\de\La^{\de}(\hb)\z^{2\de}\exp\{-\La^{\de}(\hb)\z^{2\de}\}
+2(\de-1)\exp\{-\La^{\de}(\hb)\z^{2\de}\}\\
&
\hspace{5.5cm}
-2\de\La^{\de-1}(\hb)\z^{2\de}
+2(\de-1)\La^{-1}(\hb)
].
\end{align*}

Similar arguments  can be applied to a more general  function $A$ satisfying 
the following
\begin{assumption}
\label{a-expo-gen}
Consider a function $A$ satisfying Assumption \ref{primary-assume} and such that
\be
A(x)=Ce^{-|x|^{\delta}}+r(x)\quad\text{for}\quad|x|\geq1
\ee
where $C$ is a positive constant and $r$ is a function satisfying
Assumption \ref{primary-assume} along with the asymptotics as
\begin{itemize}
\item
$r(x)=o(e^{-|x|^{\delta}})$ as $x\to\pm\infty$
\item
$r'(x)=o(|x|^{\delta-1}e^{-|x|^{\delta}})$ as $x\to\pm\infty$
\item
$r''(x)=o(|x|^{2\delta-2}e^{-|x|^{\delta}})$ as $x\to\pm\infty$.
\end{itemize}
\end{assumption}

This result above implies that for each $\hb$, the integral in 
(\ref{hbar-total-var}) converges; furthermore the convergence 
is uniform in $a(\hb)$. As in the rational case, this means that 
a variation of Theorem \ref{main-thm} can be applied to guarantee 
the existence of approximate functions in this case as well. 
The same analysis as in the previous subsection leads to a corollary 
about the EVs that lie close to $0$ (here we use once again the 
symbolism of (\ref{wkb-a-ev}), (\ref{wkb-mu-ev}) but with an $A$ 
satisfying Assumption \ref{a-expo-gen}.

\begin{corollary}
Let $b>0$ (independent of $\hb$) and consider a function $A$
satisfying Assumption \ref{a-expo-gen}. 
Then for every non-negative integer $n$ such that 
$\pi(n+\frac{1}{2})\hbar$ belongs to 
$(0,\tfrac{\pi}{2}\al_0^2(\hb))$
there exists a unique eigenvalue $i\m_{n}(\hbar)$ satisfying
\begin{equation}\nn
|\m_n(\hb)-\m_n^{WKB}(\hb)|=
\begin{cases}
\asympt\big(\frac{\hb^{5/3}}{\log\hb}\big),
\quad\text{if}\quad0<\delta<1\\
\asympt(\hbar^{5/3}),\quad\text{if}\quad\delta\geq1
\end{cases}
\quad\text{as}\quad\hb\downarrow0
\end{equation}
uniformly for $\mu_n(\hb)$ in $[\hb^b,A_{max}]$.
\end{corollary}
\begin{proof}
Using Theorem \ref{boso-near-0} (applied to our current case), 
(\ref{a-expo-gen}) and (\ref{fi-prime}) we get
\begin{align*}
|\m_n(\hb)-\m_n^{WKB}(\hb)|
& =
|A(a_n(\hb))-A(a_n^{WKB}(\hb))|\\
& \leq
C|A'(a)|[\fisf'(a)]^{-1}\hb^{\frac{5}{3}}
\quad(\text{for some $a$ between}\hspace{2pt}a_n(\hb),a_n^{WKB}(\hb))\\
& \leq
C\cdot
\begin{cases}
a^{\delta-1}e^{-a^\delta}a^{-2\delta+1}e^{a^\delta}\hb^{\frac{5}{3}},
\quad\text{if}\quad0<\delta<1\\
a^{\delta-1}e^{-a^\delta}a^{-\delta+1}e^{a^\delta}\hb^{\frac{5}{3}},
\quad\text{if}\quad\delta\geq1
\end{cases}\\
& \leq
C\cdot
\begin{cases}
\frac{\hb^{5/3}}{\log\hb^{-b}},
\quad\text{if}\quad0<\delta<1
\quad(\text{since}\hspace{4pt} a^{-\delta}\sim\frac{1}{\log\hb^{-b}})\\
\hbar^{5/3},
\quad\text{if}\quad\delta\geq1
\end{cases}
\end{align*}
where as usual $C$ denotes a generic constant. 
\end{proof} 

\begin{remark}
Having reached close enough to 0, at a distance $\hb^b$ with $b>1$, 
it is possible to show that even in the remaining interval 
$(0, i\hb^b]\subset\mathbb{C}$ the absolute difference 
$|\lam_n(\hb)-\lam_n^{WKB}(\hb)|$ is bounded by $\hb^b$, 
where $\lam_n(\hb)=i\m_n(\hb)$ and 
$\lam_n^{WKB}(\hb)=i\mu_n^{WKB}(\hb)$. 
The argument relies on the fact that there exists a very accurate 
semiclassical estimate of the total number of EVs due to Klaus 
\& Shaw (see e.g. section VI of \cite{fujii+kamvi}). Since neighboring 
Bohr-Sommerfeld approximations are at distance $O(\hb)$ from each 
other asymptotically, it follows that there is at most one such in 
the  interval $(0, i\hb^b]$. Because of the previous corollaries 
and the established 1-1 correspondence in $(\hb^b,A_{max})$, it is 
clear that there is also at most one EV in the  interval $(0, i\hb^b]$ 
and the absolute difference $|\lam_n(\hb)-\lam_n^{WKB}(\hb)|$ is indeed 
bounded by $\hb^b$.
\end{remark}

\begin{remark}
In \cite{hk},
we generalize the above to the case of several maxima and minima, and for somewhat more general asymtptotics 
at infinity. For our particular case, with bell-like potential, it follows that
the results of this section  also hold under the following assumption.
\begin{assumption}
\label{near-0-evs-potential-assume}
Suppose  there are real positive numbers $1<r^+ \leq s^+$,  so that
$$
\frac{C_1^+(x)}{ |x|^{s^+}} \leq A(x) \leq\frac{C_2^+(x)}{ |x|^{r^+}}\quad\text{for}\quad x>0
$$
where $C_1^+, C_2^+$ are bounded functions and $2 r^+ - s^+ > \frac{1}{3}$; and
there are real positive numbers $1<r^- \leq s^-$,  so that
$$
\frac{C_1^-(x)}{ |x|^{s^-}} \leq A(x) \leq\frac{C_2^-(x)}{ |x|^{r^-}}\quad\text{for}\quad x<0
$$
where $C_1^-, C_2^-$ are bounded functions and $2 r^- - s^- > \frac{1}{3}$.
Alternatively,  suppose there are real positive numbers $0<r\leq s$ so that
$$
C_1(x) e^{-|x|^s} \leq A(x) \leq C_2(x) e^{-|x|^r},\quad x\in\R
$$
where $C_1, C_2$ are bounded functions.
\end{assumption}

\end{remark}

\section{Scattering Coefficients}
\label{scattering}

In this section we will consider the scattering coefficients
for our Dirac operator (\ref{dirac})
\footnote{In this part, we work using as  guide ideas presented in
section IV of \cite{fujii+kamvi}.}
. 
As mentioned in \S \ref{from-schrodi-to-dirac},
the continuous spectrum of our Dirac operator is the whole real line.
So in this section we are considering $\lambda\in\R$. 

We begin with the case where this $\lambda$ is
\textit{idependent} of $\hb$. Under the change of variables
\be\nn
y_{\pm}=
\frac{u_{2}\pm u_{1}}{\sqrt{A\pm i\lambda}}
\ee
equation (\ref{ev-problem}) -with the help of (\ref{dirac})- is 
transformed to the following two independent equations
\be\nn
y_{\pm}''(x,\lam,\hb)=
\bigg\{
\hbar^{-2}[-A^2(x)-\lambda^2]+
\frac{3}{4}\Big[\frac{A'(x)}{A(x)\pm i\lam}\Big]^2-
\frac{1}{2}\frac{A''(x)}{A(x)\pm i\lam}
\bigg\}
y_{\pm}(x,\lam,\hb).
\ee

Again we only consider the lower index, so from now on we drop all 
the indices and work with the equation
\be\label{final-scattter-eq}
\frac{d^2y}{dx^2}=[-\hbar^{-2}\ti f(x,\lam)+\ti g(x,\lam)]y
\ee
where $\ti f$ and $\ti g$ satisfy
\be\label{f-tilde}
\ti f(x,\lam)=A^2(x)+\lam^2
\ee
and
\be\label{g-tilde}
\ti g(x,\lam)=
\frac{3}{4}\Big[\frac{A'(x)}{A(x)-i\lam}\Big]^2-
\frac{1}{2}\frac{A''(x)}{A(x)-i\lam}.
\ee

We have the following definitions.
\begin{definition}
We define the \textbf{error-control function} $\ti H(x,\lam)$ of equation 
(\ref{final-scattter-eq}) to be a primitive of
\be\label{error-control-scattering}
\frac{1}{\ti f^{\frac{1}{4}}(x,\lam)}
\frac{\partial^2}{\partial x^2}
\Big(\frac{1}{\ti f^{\frac{1}{4}}}\Big)(x,\lam)
-\frac{\ti g(x,\lam)}{\ti f^\frac{1}{2}(x,\lam)}.
\ee 
\end{definition}
\begin{definition}
For $x_1<x_2$, where $x_1\in[0,+\infty)$ and 
$x_2\in(0,+\infty)\cup\{+\infty\}$, we define
the \textbf{variation} of $\ti H$ in the interval $(x_1,x_2)$ to be
\be\label{var-control-scatt}
\mathfrak{V}_{x_1,x_2}[\ti H](\lam)=
\int_{x_1}^{x_2}
\bigg|\frac{1}{\ti f^{\frac{1}{4}}(t,\lam)}
\frac{\partial^2}{\partial t^2}
\Big(\frac{1}{\ti f^{\frac{1}{4}}}\Big)(t,\lam)
-\frac{\ti g(t,\lam)}{\ti f^\frac{1}{2}(t,\lam)}\bigg|dt.
\ee
\end{definition}

Observe that (\ref{f-tilde}) implies $\ti f>0$ in $\R$. Consequently, 
equation (\ref{final-scattter-eq}) has no turning points. Furthermore, 
notice that 
\begin{itemize}
\item
$\ti g$ is complex-valued
\item
$\ti f$ is twice continuously differentiable with respect to $x$ 
(a fact that comes from the properties of $A$ found in 
\S\ref{assume}) and
\item
$\ti g$ is continuous.
\end{itemize}
These properties allow one (cf. Theorem $2.2$ of \S $2.4$ 
from chapter 6 of \cite{olver1997}, along with the remarks from \S 5.1 
of the same chapter) to state that for $x$ in the (finite or infinite) 
interval $(x_1, x_2)\subseteq\R$ and $\kappa$ an arbitrary 
finite or infinite point in the closure of $(x_1, x_2)$, equation 
(\ref{final-scattter-eq}) has twice continuously differentiable solutions
\footnote{Since $\ti g$ is not real, we cannot expect these solutions to 
be complex conjugates.}
$w_\pm$ with
\be\nn
w_\pm(x,\hb)=
\ti f^{-\frac{1}{4}}(x,\lam)
\exp\Big\{\pm\frac{i}{\hb}\int\ti f^{\frac{1}{2}}(t,\lam)dt\Big\}
(1+\epsilon_\pm(x,\hb))
\ee
where
\be\label{error-term-scattering}
|\epsilon_\pm(x,\hb)|\hspace{5pt},\hspace{5pt}
\hb\ti f^{-\frac{1}{2}}(x,\hb)
\Big|\frac{\partial \epsilon_\pm}{\partial x}(x,\hb)\Big|
\leq
\exp\{\hb\mathfrak{V}_{\kappa,x}[\ti H](\lam)\}-1
\ee
provided that $\mathfrak{V}_{\kappa,x}(\ti H)<+\infty$ . As usual, 
the symbol $\int\ti f^{\frac{1}{2}}(t,\lam)dt$ denotes any primitive 
of $\ti f^{\frac{1}{2}}(t,\lam)$. It follows that
\begin{itemize}
\item
$\epsilon_\pm(x;\hb)\rightarrow0$ as $x\to\kappa$ and
\item
$\hb\ti f^{-\frac{1}{2}}(x,\hb)
\frac{\partial \epsilon_\pm}{\partial x}(x,\hb)\rightarrow0$ as 
$x\to\kappa$.
\end{itemize}

Notice from (\ref{error-control-scattering}) that $\ti H$ is independent
of $\hb$ whence the right-hand side of (\ref{error-term-scattering}) is 
$\asympt(\hb)$ as $\hb\downarrow0$ and fixed $x$. But 
$\mathfrak{V}_{x_1,x_2}[\ti H](\lam)<+\infty$ which implies that 
this $\asympt$-term is uniform with respect to $x$ since
$\mathfrak{V}_{\kappa,x}[\ti H](\lam)\leq
\mathfrak{V}_{x_1,x_2}[\ti H](\lam)$.
Hence 
\be\nn
w_\pm(x,\hb)\sim
\ti f^{-\frac{1}{4}}(x,\lam)
\exp\Big\{\pm\frac{i}{\hb}\int\ti f^{\frac{1}{2}}(t,\lam)dt\Big\}
\quad\text{as}\quad
\hb\downarrow0
\ee
uniformly in $(x_1, x_2)$.

Next we define the \textit{Jost solutions}. Equation 
(\ref{final-scattter-eq}) can be put in the form
\be\nn
-\frac{d^2y}{dx^2}+[-\hbar^{-2}A^2(x)+\ti g(x,\lam)]y=
\Big(\frac{\lambda}{\hb}\Big)^2y.
\ee
This is the Schr\"odinger equation with momentum $\frac{\lambda}{\hb}$, 
energy $(\frac{\lambda}{\hb})^2$ and a complex potential.
The Jost solutions are defined as the components of the 
bases $\{J_-^l, J_+^l\}$ and $\{J_-^r, J_+^r\}$ of the two-dimensional 
linear space of solutions of equation (\ref{final-scattter-eq}), 
which satisfy the asymptotic conditions
\begin{align*}
J_\pm^l(x,\lam) &\sim \exp\big\{\pm i\frac{\lambda}{\hb}x\big\}
\quad\text{as}\quad x\to-\infty\\
J_\pm^r(x,\lam) &\sim \exp\big\{\pm i\frac{\lambda}{\hb}x\big\}
\quad\text{as}\quad x\to+\infty.
\end{align*}

From scattering theory, we know that the \textit{reflection} 
$R(\lam,\hb)$ and \textit{transmition} $T(\lam,\hb)$ 
\textit{coefficients} for the waves incident on the potential from 
the right, can be expressed in terms of wronskians of the Jost 
solutions. More presicely, we have
\begin{align}
\label{r}
R(\lam,\hb) &= \frac{\W[J_-^l, J_-^r]}{\W[J_+^r, J_-^l]}\\
\label{t}
T(\lam,\hb) &= \frac{\W[J_+^r, J_-^r]}{\W[J_+^r, J_-^l]}.
\end{align}

The next step is to construct the Jost solutions as 
\textit{WKB solutions}. For this, we define the following four 
WKB solutions
\begin{align*}
\bar{w}_\pm^l(x,\hb)=
\ti f^{-\frac{1}{4}}(x,\lam)
\exp\Big\{\pm\frac{i}{\hb}\Big(\lam x+\int_{-\infty}^x 
[\ti f^{\frac{1}{2}}(t,\lam)-\lam]dt\Big)\Big\}
(1+\bar{\epsilon}_\pm^l(x,\hb))\\
\bar{w}_\pm^r(x,\hb)=
\ti f^{-\frac{1}{4}}(x,\lam)
\exp\Big\{\pm\frac{i}{\hb}\Big(\lam x+\int_{+\infty}^x 
[\ti f^{\frac{1}{2}}(t,\lam)-\lam]dt\Big)\Big\}
(1+\bar{\epsilon}_\pm^r(x,\hb))
\end{align*} 
which we are going to modify slightly in a while. If we take the 
limits as $x\to\pm\infty$ of the above, we instantly notice the
following relations between $\bar{w}_\pm^l, \bar{w}_\pm^l$ and 
the Jost solutions $J_\pm^l, J_\pm^r$; we have
\begin{align*}
J_\pm^l &=\lam^{\frac{1}{2}}\bar{w}_\pm^l\\
J_\pm^r &=\lam^{\frac{1}{2}}\bar{w}_\pm^r.
\end{align*}
Let now $w_\pm^l, w_\pm^r$ be four WKB solutions satisfying
\begin{align}
\label{wkb-left}
w_\pm^l(x,\hb)=
\ti f^{-\frac{1}{4}}(x,\lam)
\exp\Big\{\pm\frac{i}{\hb}\int_{0}^x 
\ti f^{\frac{1}{2}}(t,\lam)dt\Big\}
(1+\epsilon_\pm^l(x,\hb))
\\
\label{wkb-right}
w_\pm^r(x,\hb)=
\ti f^{-\frac{1}{4}}(x,\lam)
\exp\Big\{\pm\frac{i}{\hb}\int_{0}^x 
\ti f^{\frac{1}{2}}(t,\lam)dt\Big\}
(1+\epsilon_\pm^r(x,\hb)).
\end{align} 
Once again, the connnection between $w_\pm^l, w_\pm^r$ and 
$\bar{w}_\pm^l, \bar{w}_\pm^r$ is evident. It is
\begin{align*}
\bar{w}_\pm^l &=
\exp\Big\{\pm\frac{i}{\hb}\int_{-\infty}^0 
[\ti f^{\frac{1}{2}}(t,\lam)-\lam]dt\Big\}
w_\pm^l\\
\bar{w}_\pm^r &= 
\exp\Big\{\mp\frac{i}{\hb}\int_0^{+\infty} 
[\ti f^{\frac{1}{2}}(t,\lam)-\lam]dt\Big\}
w_\pm^r.
\end{align*}
Subsequently, for the Jost solutions we have
\begin{align}
\label{j-w-left}
J_\pm^l &=
\lam^{\frac{1}{2}}
\exp\Big\{\pm\frac{i}{\hb}\int_{-\infty}^0 
[\ti f^{\frac{1}{2}}(t,\lam)-\lam]dt\Big\}
w_\pm^l\\
\label{right-j-w}
J_\pm^r &=
\lam^{\frac{1}{2}}
\exp\Big\{\mp\frac{i}{\hb}\int_0^{+\infty} 
[\ti f^{\frac{1}{2}}(t,\lam)-\lam]dt\Big\}
w_\pm^r.
\end{align} 

Remember from \S\ref{assume} that the properties of $A$ show that the
function $t\mapsto\ti f^{\frac{1}{2}}(t,\lam)-\lam$ is in $L^1(\R)$.
Furthermore, we have
\be\label{scatter-norm}
\int_{-\infty}^0[\ti f^{\frac{1}{2}}(t,\lam)-\lam]dt=
\int_0^{+\infty}[\ti f^{\frac{1}{2}}(t,\lam)-\lam]dt=
\frac{1}{2}
\|\ti f^{\frac{1}{2}}(\cdot,\lam)-\lam\|_{L^1(\R)}
\ee
and we define
\be
\label{sigma-norm}
\sigma(\lambda):=\|\ti f^{\frac{1}{2}}(\cdot,\lam)-\lam\|_{L^1(\R)}
\ee
Substituting (\ref{j-w-left}), (\ref{right-j-w}), (\ref{scatter-norm}) 
and (\ref{sigma-norm}) in (\ref{r}), (\ref{t}) and using
\be\nn 
\W[J_+^r, J_-^r]=-2i\frac{\lam}{\hb}  
\ee
we have
\begin{align}
\label{refle}
R(\lam,\hb) &=
e^{i\frac{\sigma(\lambda)}{\hb}}
\frac{\W[w_-^l, w_-^r]}{\W[w_+^r, w_-^l]}
\\
\label{transmi}
T(\lam,\hb) &=
-\frac{2i}{\hb}
e^{i\frac{\sigma(\lambda)}{\hb}}
\frac{1}{\W[w_+^r, w_-^l]}.
\end{align}

Finally, using (\ref{wkb-left}), (\ref{wkb-right}) and 
(\ref{error-term-scattering}) we find that
\begin{itemize}
\item
$\W[w_-^l, w_-^r]=\asympt(1)$ as $\hb\downarrow0$ and
\item
$\W[w_+^r, w_-^l]=-\frac{2i}{\hb}[1+\asympt(\hb)]$ as $\hb\downarrow0$.
\end{itemize}
Substituting these last results in (\ref{refle}), (\ref{transmi})
we get that
\begin{align*}
R(\lam,\hb) &=
\frac{i\hb}{2}
e^{i\frac{\sigma(\lambda)}{\hb}}
\asympt(1)
\quad\text{as}\quad
\hb\downarrow0
\\
T(\lam,\hb) &=
e^{i\frac{\sigma(\lambda)}{\hb}}
[1+\asympt(\hb)]
\quad\text{as}\quad
\hb\downarrow0.
\end{align*}

Hence we have just proved that
\begin{theorem}
Let $A$ satisfy the assumptions of \S\ref{assume} and define $\sigma$ by 
(\ref{sigma-norm}). The scattering coefficients of equation 
(\ref{final-scattter-eq}) as defined by (\ref{r}) and 
(\ref{t}) respectively, satisfy
\begin{align}
R(\lam,\hb) &=
\asympt(\hb)
\quad\text{as}\quad
\hb\downarrow0
\\
| T(\lam,\hb) |&= 1+\asympt(\hb)
\quad\text{as}\quad
\hb\downarrow0
\end{align}
uniformly for $|\lambda|\geq\delta>0$.  
\end{theorem}

Now we turn to the case where $\lambda$ depends on $\hb$ 
($\lam=\lam(\hb)$) and particularly we let $\lambda$ approach $0$ 
like $\hb^b$ for an $\hb$-independent positive constant $b$. 
Using (\ref{f-tilde}) and (\ref{g-tilde}), we see that 
(\ref{error-control-scattering}) can be written as
\begin{align*}
\frac{1}{\ti f^{\frac{1}{4}}(x,\lam)}
\frac{\partial^2}{\partial x^2}
\Big(\frac{1}{\ti f^{\frac{1}{4}}}\Big)(x,\lam)
-\frac{\ti g(x,\lam)}{\ti f^\frac{1}{2}(x,\lam)}
&=
\frac{5}{4}[A^2(x)+\lam^2]^{-\frac{5}{2}}A(x)^2A'(x)^2\\
&\quad-
\frac{1}{2}[A^2(x)+\lam^2]^{-\frac{3}{2}}A'(x)^2\\
&\quad-
\frac{1}{2}[A^2(x)+\lam^2]^{-\frac{3}{2}}A(x)A''(x)\\
&\quad-
\frac{3}{4}[A^2(x)+\lam^2]^{-\frac{1}{2}}
\bigg[\frac{A'(x)}{A(x)-i\lam}\bigg]^2\\
&\quad+
\frac{1}{2}[A^2(x)+\lam^2]^{-\frac{1}{2}}
\frac{A''(x)}{A(x)-i\lam}.
\end{align*}
Now consider a number $q>1$. We have the following
\begin{align*}
A^2(x)+\lam^2
&=
[A^2(x)+\lam^2]^{\frac{1}{q}}[A^2(x)+\lam^2]^{1-\frac{1}{q}}\\
&>
A(x)^{\frac{2}{q}}[2A(x)\lam]^{1-\frac{1}{q}}\\
&=
2^{1-\frac{1}{q}}A(x)^{1+\frac{1}{q}}\lam^{\frac{q-1}{q}}
\end{align*}
The two estimates above show that the variation in 
(\ref{var-control-scatt}) behaves like
\be\nn 
\mathfrak{V}_{0,+\infty}[\ti H](\lam(\hb))=
\asympt\Big(\hb^{-\frac{5b(q-1)}{2}}\Big)
\quad\text{as}\quad\hb\downarrow0. 
\ee
Hence for $b>0$, $s>0$ and
$\lambda(\hb)\in[\hb^b,+\infty)$, in use of (\ref{wkb-left}) and 
(\ref{wkb-right}) we get
\begin{itemize}
\item
$\W[w_-^l, w_-^r]=\asympt\Big(\hb^{-sb}\Big)$ 
as $\hb\downarrow0$ and
\item
$\W[w_+^r, w_-^l]=
-\frac{2i}{\hb}\Big[1+\asympt\Big(\hb^{-sb}\Big)\Big]$ 
as $\hb\downarrow0$.
\end{itemize}
Substituting these last results in (\ref{refle}), (\ref{transmi})
we finally obtain that
\begin{align*}
R(\lam(\hb),\hb) &=
\frac{i\hb}{2}
e^{i\frac{\sigma(\lambda(\hb))}{\hb}}
\asympt\Big(\hb^{-sb}\Big)
\quad\text{as}\quad
\hb\downarrow0
\\
T(\lam(\hb),\hb) &=
e^{i\frac{\sigma(\lambda(\hb))}{\hb}}
\Big[1+\asympt\Big(\hb^{1-sb}\Big)\Big]
\quad\text{as}\quad
\hb\downarrow0.
\end{align*}

So, we have  showed the following
\begin{theorem}
\label{spectrum+scatter}
Let $A$ satisfy the assumptions of \S\ref{assume}. Consider $b,s>0$ 
(independent of $\hb$). Then the reflection coefficient and the transmission 
coefficient of equation (\ref{final-scattter-eq}) as defined by (\ref{r}) 
and (\ref{t}) respectively, satisfy
\begin{align}
\label{final-r}
R(\lam(\hb),\hb) &=
\asympt\Big(\hb^{1-sb}\Big)
\quad\text{as}\quad
\hb\downarrow0
\\
\label{final-t}
|T(\lam(\hb),\hb)| &=
1+\asympt\Big(\hb^{1-sb}\Big)
\quad\text{as}\quad
\hb\downarrow0
\end{align}
uniformly for $\lambda(\hb)$ in any closed interval of $[\hb^b,+\infty)$.  
\end{theorem}
\begin{remark}
We can ensure that $b$ is as large  as we want by letting $s$ very small 
if we are happy with a weak error estimate $O(\hb^{\epsilon})$ for small 
positive $\epsilon$, as $\hb\downarrow0$. We can at best guarantee 
asymptotics of order $O(\hb^{1-\epsilon})$ for small positive 
$\epsilon$, if we are allowed to accept a small $b$.
\end{remark}
\begin{remark}
The results provided by Theorems 2.1 and
2.4 in \cite{fujii+kamvi}, where the potential is real-analytic,
actually imply exponential decay as $\hb\downarrow0$. Still our own 
estimate here is good enough for the applications to the theory of 
focusing NLS.
\end{remark}
\begin{remark}
We check that
\be\nn
|T(\lam(\hb),\hb)|^2-|R(\lam(\hb),\hb)|^2=
1+\asympt\Big(\hb^{1-sb}\Big)
\quad\text{as}\quad
\hb\downarrow0
\ee
as of course it should be the case.
\end{remark}

\section{Conclusion}
\label{conclusion}

The results in the last three paragraphs are stronger than those 
of \cite{fujii+kamvi} in the sense that they cover analytic 
bell-shaped potentials as well as non-analytic potentials with a 
certain smoothness. The reflection coefficient estimate is weaker 
but this does not affect the results and proofs pertaining to the 
applications to the semiclassical limit of the NLS equation. 
On the other hand, the more important Bohr-Sommerfeld estimate is 
stronger. We refer to \cite{fujii+kamvi} for the actual statements
of the precise results; the detailed proofs are now much more
straightforward. Indeed, the proof of Proposition 6.1 in 
\cite{fujii+kamvi} is now trivial in view of the uniform validity of 
the Bohr-Sommerfeld estimate near 0 and the remaining propositions 
in paragraph 6 (statements and proofs) are unchanged.

\appendix

\section{Airy Functions}
\label{airy_functions}

In this section, some basic properties of \textit{Airy functions} are 
presented. For further reading one may consult \cite{olver1997}.

\begin{figure}[H]
\includegraphics[scale=0.4]{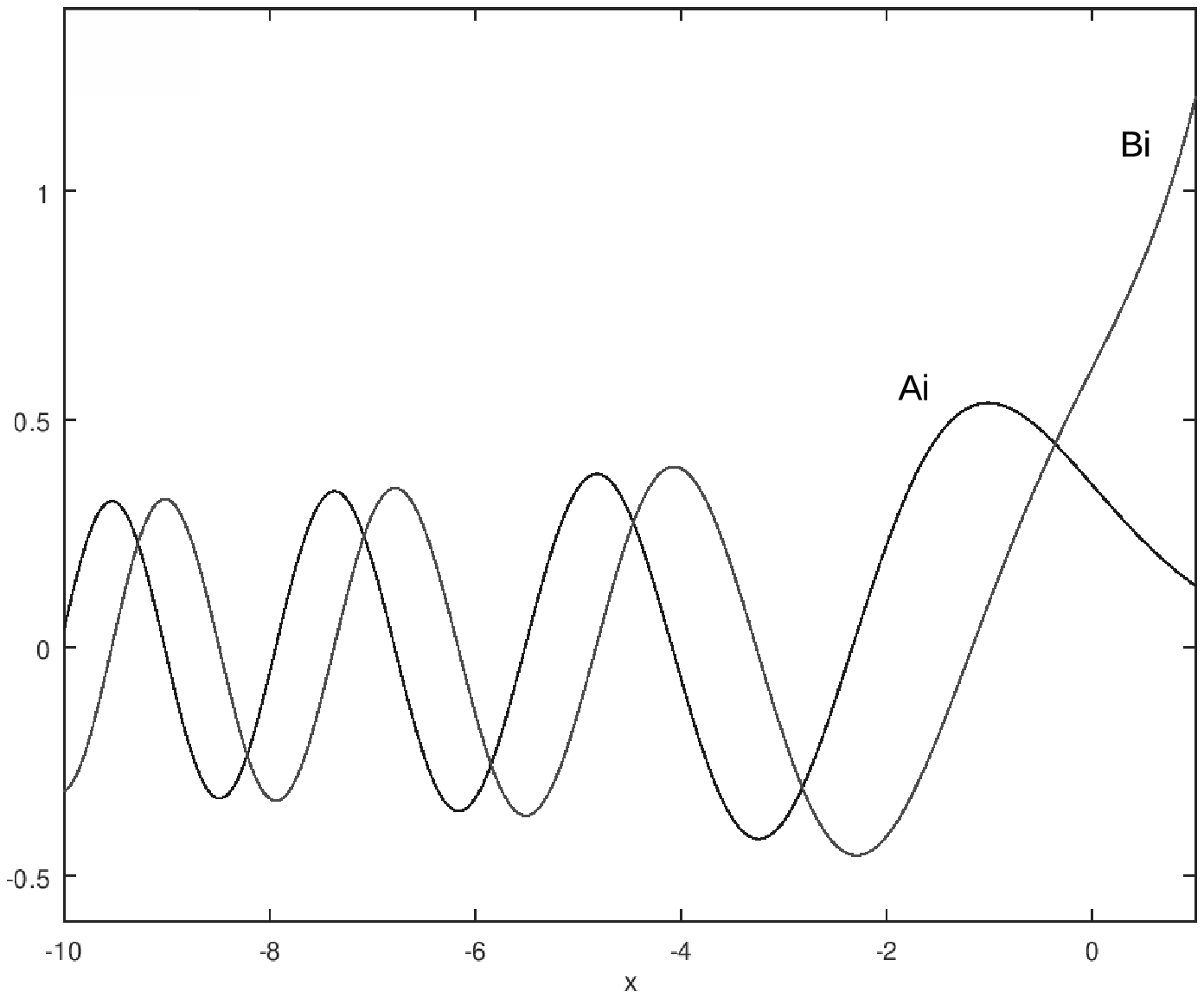}
\caption{The Airy functions $Ai$, $Bi$ on the real line.}
\end{figure}
\noindent
Consider the \textit{Airy equation}
\begin{equation*}
-\dfrac{d^{2}w}{dt^{2}}+tw=0,\quad t\in\mathbb{R}
\end{equation*}
We denote by $Ai$ and $Bi$ its two linearly independent solutions having 
the asymptotics
\begin{equation}\label{airy_1}
Ai(t)=
\frac{1}{2\sqrt{\pi}}
t^{-\frac{1}{4}}\exp\{-\tfrac{2}{3}t^\frac{3}{2}\}
\big[1+O\big(t^{-\frac{3}{2}}\big)\big]
\quad\text{as}\quad t\rightarrow+\infty
\end{equation}
and
\begin{equation}\label{airy_2}
Bi(t)=
-\frac{1}{\sqrt{\pi}}
|t|^{-\frac{1}{4}}\sin\Big(\tfrac{2}{3}|t|^\frac{3}{2}-
\tfrac{\pi}{4}\Big)+O\big(|t|^{-\frac{7}{4}}\big)
\quad\text{as}\quad t\rightarrow-\infty
\end{equation}
\noindent
Their behavior on the opposite side of the real line is known to be
\begin{equation}\label{airy_3}
Ai(t)=
\frac{1}{\sqrt{\pi}}
|t|^{-\frac{1}{4}}\sin\Big(\tfrac{2}{3}|t|^{\frac{3}{2}}+
\tfrac{\pi}{4}\Big)+O\big(|t|^{-\frac{7}{4}}\big)
\quad\text{as}\quad t\rightarrow-\infty
\end{equation}
and
\begin{equation*}
Bi(t)\leq C(1+t)^{-\frac{1}{4}}\exp\{\tfrac{2}{3}t^\frac{3}{2}\},
\quad t\geq0
\end{equation*}
where $C$ is a positive constant.
\noindent
Observe that as $t\to-\infty$, $Ai$ and $Bi$ only differ by a phase shift. 
Also $Ai(t),Bi(t)>0$ for all $t\geq0$. Note that all asymptotic relations 
(\ref{airy_1}), (\ref{airy_2}) and (\ref{airy_3}) can be differentiated 
in $t$; for example
\begin{equation*}
Ai'(t)=
-\frac{1}{\sqrt{\pi}}
|t|^\frac{1}{4}\cos\Big(\tfrac{2}{3}|t|^{\frac{3}{2}}+
\tfrac{\pi}{4}\Big)+
O\big(|t|^{-\frac{5}{4}}\big)
\quad\text{as}\quad t\rightarrow-\infty
\end{equation*}
and
\begin{equation*}
Ai'(t)=
-\frac{1}{2\sqrt{\pi}}
t^\frac{1}{4}\exp\{-\tfrac{2}{3}t^\frac{3}{2}\}
\big[1+O\big(t^{-\frac{3}{2}}\big)\big]
\quad\text{as}\quad t\rightarrow+\infty.
\end{equation*}

Another property says that
\begin{equation}\nn
|Ai(t)|\leq C(1+|t|)^{-\frac{1}{4}},\qquad t\in\mathbb{R}
\end{equation} 
where $C$ is a positive constant. The wronskian of $Ai$, $Bi$ satisfies
\begin{equation*}
\mathcal{W}\big[Ai,Bi\big](t):=Ai(t)Bi'(t)-Ai'(t)Bi(t)=
\frac{1}{\pi},
\quad t\in\mathbb{R}.
\end{equation*}

In order to have a convenient way of assessing the magnitudes of $Ai$ and 
$Bi$ we introduce a \textit{modulus function} $M$, a \textit{phase function} 
$\thv$ and a \textit{weight function} $E$ related by
\begin{equation*}
E(x)Ai(x)=M(x)\sin\thv(x),\quad \frac{1}{E(x)}Bi(x)=M(x)\cos\thv(x),
\quad x\in\mathbb{R}.
\end{equation*}
Actually, we choose $E$ as follows. Denote by $c_{*}$ the biggest 
negative root of the equation $Ai(x)=Bi(x)$ 
(numerical calculations show that $c_{*}=-0.36605$ correct up to five 
decimal places); then define
\begin{equation}\nn
E(x)=
\begin{cases}
1,\quad x\leq c_{*}\\
\sqrt{\frac{Bi(x)}{Ai(x)}},\quad x>c_{*}
\end{cases}
\end{equation}
With this choice in mind, $M$, $\theta$ become
\begin{equation}\nn
M(x)=
\begin{cases}
\sqrt{Ai^{2}(x)+Bi^{2}(x)},\quad x\leq c_{*}\\
\sqrt{2Ai(x)Bi(x)},\quad x>c_{*}
\end{cases}\text{and}\quad
\thv(x)=
\begin{cases}
\arctan\Big[\frac{Ai(x)}{Bi(x)}\Big],\quad x\leq c_{*}\\
\frac{\pi}{4},\quad x>c_{*}
\end{cases}
\end{equation}
where the branch of the inverse tangent is continuous and equal to
$\frac{\pi}{4}$ at $x=c_{*}$. For these functions the asymptotics for 
large $|x|$ read
\begin{align*}
E(x)\sim & 
\begin{cases}
1,\quad x\rightarrow-\infty\\
\sqrt{2}\exp\{\tfrac{2}{3}t^\frac{3}{2}\},
\quad x\rightarrow+\infty
\end{cases}\\
M(x)\sim & \frac{1}{\sqrt{\pi}}
|x|^{-\frac{1}{4}},
\quad |x|\rightarrow+\infty\\
\thv(x)= & 
\begin{cases}
\frac{2}{3}
|x|^\frac{3}{2}+
\frac{\pi}{4}+
\mathcal{O}\big(\frac{3}{2}|x|^{-\frac{3}{2}}\big),
\quad x\rightarrow-\infty\\
\frac{\pi}{4},\quad x\rightarrow+\infty
\end{cases}
\end{align*}

\section{Parabolic Cylinder Functions}
\label{parabolic-cylinder-functions}

The result of the main theorem found in section \S\ref{approximate-solutions}, 
involves parabolic cylinder functions (cf. \cite{abramo+stegun}). 
So in this section we state a few properties which will be in heavy use, 
especially about their asymptotic character, wronskians and zeros. We
prove none of them. For a rigorous exposition on parabolic
cylinder functions, one may consult \S 5 of \cite{olver1975} or \S 12 of
\cite{olver-et-al} and the references therein.

Consider \textit{Weber's equation}
\be\label{pcf}
\frac{d^2 w}{dx^2}=\Big(\tfrac{1}{4}x^2+b\Big)w.
\ee
The behavior of its solutions depends on the sign of $b$. When $b$ is 
negative then there exist two turning points $\pm 2\sqrt{-b}$. The
solutions are of oscillatory type in the interval between these points
but not in the exterior intervals. When $b>0$ there are no
real turning points and there are no oscillations at all. Since only
the case $b\leq0$ will be of interest to us, from now on we seldom mention 
properties having to do with the other case.

Standard solutions of (\ref{pcf}) are $U(\pm x,b)$ and $\ol U(\pm x,b)$ 
defined by
\begin{multline*}
U(\pm x,b)=
\frac{\pi^{\frac{1}{2}}2^{-\frac{1}{4}(2b+1)}}{\G(\frac{3}{4}+\frac{1}{2}b)}
e^{-\frac{1}{4}x^2}
\Hypergeometric{1}{1}{\tfrac{1}{4}+\tfrac{1}{2}b}{\tfrac{1}{2}}{\tfrac{1}{2}x^2}\\
\mp
\frac{\pi^{\frac{1}{2}}2^{-\frac{1}{4}(2b-1)}}{\G(\frac{1}{4}+\frac{1}{2}b)}
x
e^{-\frac{1}{4}x^2}
\Hypergeometric{1}{1}{\tfrac{3}{4}+\tfrac{1}{2}b}{\tfrac{3}{2}}{\tfrac{1}{2}x^2}
\end{multline*}
\begin{multline*}
\ol U(\pm x,b)=
\pi^{-\frac{1}{2}}2^{-\frac{1}{4}(2b+1)}\G(\tfrac{1}{4}-\tfrac{1}{2}b)
\sin(\tfrac{3}{4}\pi-\tfrac{1}{2}b\pi)
e^{-\frac{1}{4}x^2}
\Hypergeometric{1}{1}{\tfrac{1}{4}+\tfrac{1}{2}b}{\tfrac{1}{2}}{\tfrac{1}{2}x^2}\\
\mp
\pi^{-\frac{1}{2}}2^{-\frac{1}{4}(2b-1)}\G(\tfrac{3}{4}-\tfrac{1}{2}b)
\sin(\tfrac{5}{4}\pi-\tfrac{1}{2}b\pi)
x
e^{-\frac{1}{4}x^2}
\Hypergeometric{1}{1}{\tfrac{3}{4}+\tfrac{1}{2}b}{\tfrac{3}{2}}{\tfrac{1}{2}x^2}
\end{multline*}
where $_1 F_1$ denotes the confluent hypergeometric function 
(again cf. \cite{abramo+stegun}). The pair
$U(x,b), \ol U(x,b)$ is a numerically satisfactory set of solutions (in the 
sense of \cite{miller1950}) when $x\geq0$ and $b\leq0$; both are continuous 
in $x$ and $b$ in this region.

For $b\in\R$, their values at $x=0$ obey
\begin{align*}
U(0,b) &=
\pi^{-\frac{1}{2}}2^{-\frac{1}{4}(2b+1)}
\G(\tfrac{1}{4}-\tfrac{1}{2}b)\sin(\tfrac{\pi}{4}-\tfrac{1}{2}b\pi)\\
U'(0,b) &= -
\pi^{-\frac{1}{2}}2^{-\frac{1}{4}(2b-1)}
\G(\tfrac{3}{4}-\tfrac{1}{2}b)\sin(\tfrac{3\pi}{4}-\tfrac{1}{2}b\pi)\\
\ol U(0,b) &=
\pi^{-\frac{1}{2}}2^{-\frac{1}{4}(2b+1)}
\G(\tfrac{1}{4}-\tfrac{1}{2}b)\sin(\tfrac{3\pi}{4}-\tfrac{1}{2}b\pi)\\
\ol U'(0,b) &= -
\pi^{-\frac{1}{2}}2^{-\frac{1}{4}(2b-1)}
\G(\tfrac{3}{4}-\tfrac{1}{2}b)\sin(\tfrac{5\pi}{4}-\tfrac{1}{2}b\pi).
\end{align*} 

Those values of $b$ that make the Gamma functions in the definitions of
$U$ and $\ol U$ infinite (the Gamma function has simple poles at the 
non-positive integers), are called \textit{exceptional values}. 
For a fixed $b\in\R$ other than an exceptional value,
the behaviors of $U$ and $\ol U$ as $x\to+\infty$ satisfy
\begin{align}
\label{U-asympt}
U(x,b) & \sim
x^{-b-\frac{1}{2}}e^{-\frac{1}{4}x^2}\\
\nn
U'(x,b) & \sim
-\frac{1}{2}x^{-b+\frac{1}{2}}e^{-\frac{1}{4}x^2}\\
\nn
\ol U(x,b) & \sim
\sqrt{\frac{2}{\pi}}
\G(\tfrac{1}{2}-b)
x^{b-\frac{1}{2}}e^{\frac{1}{4}x^2}\\
\nn
\ol U'(x,b) & \sim
\frac{1}{\sqrt{2\pi}}
\G(\tfrac{1}{2}-b)
x^{b+\frac{1}{2}}e^{\frac{1}{4}x^2}.
\end{align}
These estimates are uniform in $b$ when $b$ takes values over a fixed
compact interval not containing exceptional values.

For the wronskian of $U(\cdot,b),\ol U(\cdot,b)$ we have
\be\label{wronskian-pcf}
\W[U(\cdot,b),\ol U(\cdot,b)](x)=
\sqrt{\frac{2}{\pi}}
\G\Big(\frac{1}{2}-b\Big),
\quad x\in\R.
\ee

When $b=0$ the standard solutions of equation (\ref{pcf}) are related to
the \textit{modified Bessel functions} $K_{\frac{1}{4}}$ and 
$I_{\frac{1}{4}}$ in the following way. For $x\geq0$ we have
\begin{align*}
U(x,0) & =
\frac{1}{\sqrt{2\pi}}
x^{\frac{1}{2}}K_{\frac{1}{4}}(\tfrac{1}{4}x^2)\\
\ol U(x,0) & =
\sqrt{\pi}x^{\frac{1}{2}}I_{\frac{1}{4}}(\tfrac{1}{4}x^2)+
\frac{1}{\sqrt{2\pi}}x^{\frac{1}{2}}K_{\frac{1}{4}}(\tfrac{1}{4}x^2).
\end{align*}

In order to express the character of these standard solutions
for large (in absolute value) negative $b$, we need some preparations 
first. Take $\nu\gg1$ to be a large positive number and set
$b=-\frac{1}{2}\nu^2$ and $x=\nu y\sqrt{2}$ where $y\geq0$. If we
consider the fuction $\eta$ to be
\be\label{eta-definition}
\eta(y)=
\begin{cases}
-[\frac{3}{2}\int_{y}^{1}\sqrt{1-s^2}ds]^{\frac{2}{3}},
\quad 0\leq y\leq1\\
[\frac{3}{2}\int_{1}^{y}\sqrt{s^2-1}ds]^{\frac{2}{3}},
\quad y\geq1
\end{cases}
\ee
then as $\nu\to+\infty$ we have
\begin{align}\label{u-asymptotics}
U(\nu y\sqrt{2},-\tfrac{1}{2}\nu^2) & =
\frac{2^{\frac{1}{2}}\pi^{\frac{1}{4}}
\G(\tfrac{1}{2}+\frac{1}{2}\nu^2)^{\frac{1}{2}}}
{\nu^{\frac{1}{6}}} 
\Big(\frac{\eta}{y^2-1}\Big)^{\frac{1}{4}}
\Big[
\Ai(\nu^{\frac{4}{3}}\eta)+
\frac{M(\nu^{\frac{4}{3}}\eta)}{E(\nu^{\frac{4}{3}}\eta)}
\asympt(\nu^{-2})
\Big]
\\
\label{ubar-asymptotics}
\ol U(\nu y\sqrt{2},-\tfrac{1}{2}\nu^2) & =
\frac{2^{\frac{1}{2}}\pi^{\frac{1}{4}}
\G(\tfrac{1}{2}+\frac{1}{2}\nu^2)^{\frac{1}{2}}\eta^{\frac{1}{4}}}
{\nu^{\frac{1}{6}}(y^2-1)^{\frac{1}{4}}} 
\Big[
\Bi(\nu^{\frac{4}{3}}\eta)+
M(\nu^{\frac{4}{3}}\eta)E(\nu^{\frac{4}{3}}\eta)
\asympt(\nu^{-2})
\Big]
\end{align}
where $\Ai$, $\Bi$, $E$ and $M$ are the standard Airy functions'
terminology (cf. section \ref{airy_functions} in the appendix).

For $b\leq0$, the number of zeros of $U(\cdot,b)$ in the interval 
$[0,+\infty)$ is $\floor{\tfrac{1}{4}-\frac{1}{2}b}$ while
$\ol U(\cdot,b)$ has $\floor{\tfrac{3}{4}-\frac{1}{2}b}$ zeros
in $[0,+\infty)$. Actually, the zeros of $U(\cdot,b)$ and 
$\ol U(\cdot,b)$ do not cross each other. They interlace, with the 
largest one belonging to $\ol U(\cdot,b)$. For sufficiently large 
$|b|$, all the real zeros of these two functions lie to the left 
of $2\sqrt{-b}$, the positive turning point of Weber's equation
\footnote{
For $U(\cdot,b)$, this result holds for all $b\leq0$.
}
. 

\begin{figure}[H]
\centering
\includegraphics[scale=1.2]{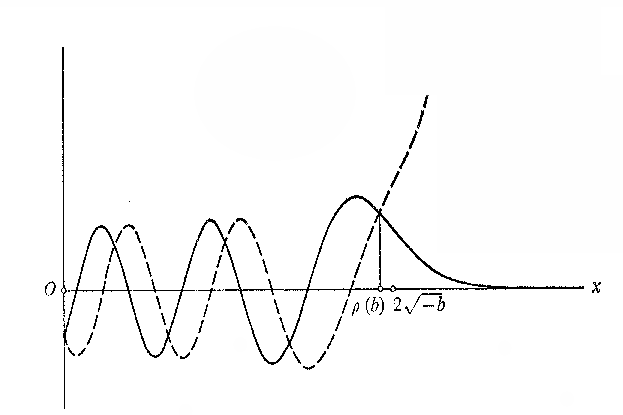}
\caption{An example of Parabolic Cylinder Functions $U(\cdot;b)$ (continuous) 
and $\ol U(\cdot;b)$ (dashed) for some $b<0$.}
\end{figure}

To express the errors for the approximations of our problem, we need
to define some auxiliary functions having to do with the nature of 
$U(\cdot,b)$ and $\ol U(\cdot,b)$ for negative $b$. In this case the 
character of each is partly oscillatory and partly exponential, so we 
introduce one weight function $\esf$, two modulus functions $\msf$ 
and $\nsf$, and finally two phase functions $\8$ and $\om$.

We denote by $\ro(b)$ the largest real root of the equation
\be\nn
U(x,b)=\ol U(x,b).
\ee
We know (cf. \S 13 of \cite{olver-et-al} and the references therein) 
that $\ro(0)=0$ and $\ro(b)>0$ for $b<0$. Also, $\ro$ is continuous 
when $b\in(-\infty,0]$. An asymptotic estimate for large negative 
$b$ is
\be\label{largest-root-pcf-asympt}
\ro(b)=2\sqrt{-b}+c_*(-b)^{-\frac{1}{6}}+
\asympt\big(b^{-\frac{5}{6}}\big)
\quad\text{as}\quad
b\to-\infty
\ee
where $c_*$ ($\approx-0.36605$) is the smallest in absolute value root 
of the equation $\Ai(x)=\Bi(x)$. 

For $b\leq0$ we define 
\begin{equation}\nn
\esf(x,b)=
\begin{cases}
1,\quad 0\leq x\leq\ro(b)\\
\Big[\frac{\ol U(x,b)}{U(x,b)}\Big]^{1/2},\quad x>\ro(b).
\end{cases}
\end{equation}
It is seen that $\esf$ is continuous in the region 
$[0,+\infty)\times(-\infty,0]$ of the $(x,b)$-plane and for $b\leq0$ 
fixed, $\esf(\cdot,b)$ is non-decreasing in the interval $[0,+\infty)$.
Again for $b\leq0$ and $x\geq0$ we set
\be\nn
U(x,b)=\frac{1}{\esf(x,b)}\msf(x,b)\sin\8(x,b),\quad 
\ol U(x,b)=\esf(x,b)\msf(x,b)\cos\8(x,b)
\ee
and
\be\nn
U'(x,b)=\frac{1}{\esf(x,b)}\nsf(x,b)\sin\om(x,b),\quad 
\ol U'(x,b)=\esf(x,b)\nsf(x,b)\cos\om(x,b).
\ee
Thus
\be\label{m-definition}
\msf(x,b)=
\begin{cases}
\big[U(x,b)^{2}+\ol{U}(x,b)^{2}\big]^{1/2},\quad 0\leq x\leq\ro(b)\\
\big[2U(x,b)\ol U(x,b)\big]^{1/2},\quad x>\ro(b)
\end{cases}
\ee
and
\be\nn
\theta(x,b)=
\begin{cases}
\arctan\Big[\frac{U(x,b)}{\ol U(x,b)}\Big],\quad 0\leq x\leq\ro(b)\\
\frac{\pi}{4},\quad x>\ro(b)
\end{cases}
\ee
where the branch of the inverse tangent is continuous and equal to
$\frac{\pi}{4}$ at $x=\ro(b)$. 

Similarly
\be\nn
\nsf(x,b)=
\begin{cases}
\Big[U'(x,b)^{2}+\ol{U}'(x,b)^{2}\Big]^{1/2},\quad 0\leq x\leq\ro(b)\\
 \\
\bigg[\frac{U'(x,b)^{2}\ol U(x,b)^2+
           \ol U'(x,b)^{2}U(x,b)^{2}}{U(x,b)\ol U(x,b)}\bigg]^{1/2},
           \quad x>\ro(b)
\end{cases}
\ee
and
\be\nn
\omega(x,b)=
\begin{cases}
\arctan\Big[\frac{U'(x,b)}{\ol{U}'(x,b)}\Big],\quad 0\leq x\leq\ro(b)\\
 \\
\arctan\Big[\frac{U'(x,b)\ol{U}(x,b)}{\ol{U}'(x,b)U(x,b)}\Big],
\quad \quad x>\ro(b)
\end{cases}
\ee
where the branches of the inverse tangents are chosen to be continuous and 
fixed by $\omega(x,b)\rightarrow-\frac{\pi}{4}$ as $x\rightarrow+\infty$.

For large $x$ we have
\be\nn
\esf(x,b)\sim\Big(\frac{2}{\pi}\Big)^{\frac{1}{4}}
\G(\tfrac{1}{2}-b)^{\frac{1}{2}}x^{b}e^{\frac{1}{4}x^2}
\ee
and
\be\label{M,N-asymptotics}
\msf(x,b)\sim\Big(\frac{8}{\pi}\Big)^{\frac{1}{4}}
\frac{\G(\tfrac{1}{2}-b)^{\frac{1}{2}}}{x^{\frac{1}{2}}},
\quad
\nsf(x,b)\sim
\frac{\G(\tfrac{1}{2}-b)^{\frac{1}{2}}}{(2\pi)^{\frac{1}{4}}}
x^{\frac{1}{2}}.
\ee
Both of these hold for fixed $b$ and are also uniform for $b$ ranging over
any compact interval in $(-\infty,0]$.

\section{A Theorem on Integral Equations}
\label{exist-proof}

The proofs of theorems about WKB approximation when there is an absence 
of turning points (like Theorems 2.1 and 2.2 in chapter 6 of
\cite{olver1997}), may be adapted to other types of 
approximate solutions of linear differential equations 
where turning points may be present. 
For second-order equations the basic steps consist of
\begin{itemize}
\item[(i)]
construction of a \textit{Volterra integral equation} for the error 
term -say $h$- of the solution, by the method of 
\textit{variation of parameters}
\item[(ii)]
construction of  the 
\textit{Liouville-Neumann expansion} (a uniformly convergent series) 
for the solution $h$
of the integral equation in (i) by \textit{Picard's method of successive 
approximations}
\item[(iii)]
confirmation that $h$ is twice differentiable by construction of 
similar series for $h'$ and $h''$
\item[(iv)]
production of bounds for $h$ and $h'$ by majoring the 
Liouville-Neumann expansion.
\end{itemize}

It would be tedious to carry out all these steps in every case. But
we have the following general theorem which automatically provides (ii), 
(iii) and (iv) in problems relevant to us.

\begin{theorem}\label{sing-int-eq}
\footnote{
This is Theorem 10.2 found in chapter 6 of \cite{olver1997}. It is a 
variant of Theorem 10.1 from the same reference.
}
Consider the equation
\begin{equation}\label{int_eq1}
h(\z)=\int_{\beta}^{\z}\mathsf{K}(\z,t)\phi(t)\{J(t)+h(t)\}dt
\end{equation}
for the function $h$ accompanied by the following assumptions
\begin{itemize}
\item
the \say{path} of integration consists of a segment $[\beta,\zeta]$ 
of the real axis, finite or infinite where $\beta\leq t\leq\z\leq\g$
\item
the real functions $J$ and $\phi$ are continuous in $(\beta,\g)$ 
except for a finite number of discontinuities and infinities
\item
the real kernel $\mathsf{K}$ and its first two partial derivatives 
with respect to $\z$ are continuous functions of both variables when 
$\z,t\in(\beta,\g)$
\item
$\mathsf{K}(\z,\z)=0,\quad\z\in(\beta,\g)$
\item
when $\z\in(\beta,\g)$ and $t\in(\beta,\z]$ we have
\begin{equation*}
|\mathsf{K}(\z,t)|\leq P_{0}(\z)Q(t),\qquad 
\Big|\frac{\partial\mathsf{K}(\z,t)}{\partial\z}\Big|\leq P_{1}(\z)Q(t),
\qquad\Big|\frac{\partial^{2}\mathsf{K}(\z,t)}{\partial\z^{2}}\Big|
\leq P_{2}(\z)Q(t)
\end{equation*}
where the $P_{j},j=0,1,2$ and $Q$ are continuous real functions, the 
$P_{j},j=0,1,2$ being positive.
\item
when $\z\in(\beta,\g)$, the integral
\begin{equation*}
\Phi(\z)=\int_{\beta}^{\z}|\phi(t)|dt
\end{equation*}
converges and the following suprema
\begin{equation*}
\kappa=\sup_{\z\in(\beta,\g)}\{Q(\z)|J(\z)|\},\qquad
\kappa_{0}=\sup_{\z\in(\beta,\g)}\{P_{0}(\z)Q(\z)\}
\end{equation*}
are finite.
\end{itemize}
Under these assumptions, equation (\ref{int_eq1}) has a unique solution 
$h$ which is continuously differentiable in $(\beta,\g)$
and satisfies
\begin{equation*}
\frac{h(\z)}{P_{0}(\z)}\to0\qquad\frac{h'(\z)}{P_{1}(\z)}
\to0\qquad\text{as}\quad\z\downarrow\beta.
\end{equation*}
Furthermore,
\begin{equation*}
\frac{|h(\z)|}{P_{0}(\z)},\frac{|h'(\z)|}{P_{1}(\z)}\leq
\frac{\kappa}{\kappa_{0}}[\exp\{\kappa_{0}\Phi(\z)\}-1]
\end{equation*}
and $h''$ is continuous except at the discontinuities -if any- of 
$\phi,J$.
\end{theorem}
\begin{proof}
The proof is a slight variation of that for Theorem 10.1 of chapter 6 
in  \cite{olver1997}.
\end{proof}

We are going to use this theorem to prove the existence and behavior of
approximate solutions of the equation
\be\label{eq-append}
\frac{d^2\mathcal{Y}}{d\z^2}=
\big[\hb^{-2}(\z^2-\alpha^2)+\ps(\z,\hb,\al)\big]\mathcal{Y}.
\ee
We have the following
\btheo\label{thm-on-exist-int-eq}
For each value of $\hb$, assume that the function $\ps(\z,\hb,\al)$ is 
continuous in the region $[0,Z)\times[0,\delta]$ of the $(\z,\al)$-plane
\footnote{
Here $Z$ is always positive and may depend continuously on $\al$, or
be infinite. Also, $\delta$ is a positive finite constant.
}, take $\Omega$ as in (\ref{omega-asymptotics})
and consider that 
\be\nn
\var_{0,Z}[H](\al,\hb)=
\int_{0}^{Z}\frac{|\ps(t,\al)|}{\Om(t\sqrt{2\hb^{-1}})}dt
\ee
converges uniformly with 
respect to $\al$. Then in this region, equation 
(\ref{eq-append}) has 
solutions $\mathcal{Y}_1$ and $\mathcal{Y}_2$ which are continuous, have 
continuous first and 
second partial $\z$-derivatives and are given by
\bea\nn
\mathcal{Y}_1(\z,\al,\hb)=
U(\z\sqrt{2\hb^{-1}},-\tfrac{1}{2}\hb^{-1}\al^2)+
\epsilon_1 (\z,\al,\hb)\\
\nn
\mathcal{Y}_2(\z,\al,\hb)=
\ol U(\z\sqrt{2\hb^{-1}},-\tfrac{1}{2}\hb^{-1}\al^2)+
\epsilon_2 (\z,\al,\hb)
\eea
where
\begin{multline}\label{bound1}
\frac{|\epsilon_1 (\z,\al,\hb)|}
{\msf(\z\sqrt{2\hb^{-1}},-\tfrac{1}{2}\hb^{-1}\al^2)},
\frac{\Big|\frac{\partial \epsilon_1 }
{\partial\z}(\z,\al,\hb)\Big|}{\sqrt{2\hb^{-1}}
\nsf(\z\sqrt{2\hb^{-1}},-\tfrac{1}{2}\hb^{-1}\al^2)}\\
\leq
\frac{1}{\esf(\z\sqrt{2\hb^{-1}},-\tfrac{1}{2}\hb^{-1}\al^2)}
\Big(\exp\big\{\tfrac{1}{2}(\pi\hb)^{\frac{1}{2}}l(-\tfrac{1}{2}\hb^{-1}\al^2)
\mathcal{V}_{\z,Z}[H](\al,\hb)\big\}-1\Big)
\end{multline}
and
\begin{multline}\label{bound2}
\frac{|\epsilon_2 (\z,\al,\hb)|}
{\msf(\z\sqrt{2\hb^{-1}},-\tfrac{1}{2}\hb^{-1}\al^2)},
\frac{\Big|\frac{\partial \epsilon_2 }
{\partial\z}(\z,\al,\hb)\Big|}{\sqrt{2\hb^{-1}}
\nsf(\z\sqrt{2\hb^{-1}},-\tfrac{1}{2}\hb^{-1}\al^2)}\\
\leq
\esf(\z\sqrt{2\hb^{-1}},-\tfrac{1}{2}\hb^{-1}\al^2)
\Big(\exp\big\{\tfrac{1}{2}(\pi\hb)^{\frac{1}{2}}l(-\tfrac{1}{2}\hb^{-1}\al^2)
\mathcal{V}_{0,\z}[H](\al,\hb)\big\}-1\Big).
\end{multline}
\etheo
\begin{proof}
We will prove the theorem only for the first solution since the proof for
the second follows mutatis mutandis. Observe that the approximating function 
$U(\z\sqrt{2\hb^{-1}},-\tfrac{1}{2}\hb^{-1}\al^2)$ satisfies
$\frac{d^2U}{d\z^2}=\hb^{-2}(\z^2-\alpha^2)U$.  If we subtract this from
(\ref{eq-append}) we obtain the following differential equation for the 
error term
\be\nn
\frac{d^2\epsilon_1}{d\z^2}-\hb^{-2}(\z^2-\alpha^2)\epsilon_1=
\ps(\z,\al,\hb)\big[\epsilon_1+
U(\z\sqrt{2\hb^{-1}},-\tfrac{1}{2}\hb^{-1}\al^2)\big].
\ee
By use of the method of variation of parameters
and also  (\ref{wronskian-pcf}) one arrives at the integral equation
\be\nn
\epsilon_1 (\z,\al,\hb)=
\frac{1}{2}
\frac{(\pi\hb)^{\frac{1}{2}}}{\G\big(\frac{1}{2}+\frac{1}{2}\hb^{-1}\al^2\big)}
\int_{\z}^{Z}
\mathcal{K}(\z,t)
\ps(t,\al,\hb)\big[\epsilon_1(t,\al,\hb)+
U(t\sqrt{2\hb^{-1}},-\tfrac{1}{2}\hb^{-1}\al^2)\big]
dt
\ee
in which
\begin{multline*}
\mathcal{K}(\z,t)=
U(\z\sqrt{2\hb^{-1}},-\tfrac{1}{2}\hb^{-1}\al^2)
\ol U(t\sqrt{2\hb^{-1}},-\tfrac{1}{2}\hb^{-1}\al^2)\\
-U(t\sqrt{2\hb^{-1}},-\tfrac{1}{2}\hb^{-1}\al^2)
\ol U(\z\sqrt{2\hb^{-1}},-\tfrac{1}{2}\hb^{-1}\al^2).
\end{multline*}

Bounds for the \textit{kernel} $\mathcal{K}$ and its first two partial
derivatives (with respect to $\z$) are expressible in terms of the 
auxiliary functions $\esf, \msf$ and $\nsf$. We have
\begin{align*}
|\mathcal{K}(\z,t)| & \leq
\frac{\esf(t\sqrt{2\hb^{-1}},-\tfrac{1}{2}\hb^{-1}\al^2)}
{\esf(\z\sqrt{2\hb^{-1}},-\tfrac{1}{2}\hb^{-1}\al^2)}
\msf(\z\sqrt{2\hb^{-1}},-\tfrac{1}{2}\hb^{-1}\al^2)
\msf(t\sqrt{2\hb^{-1}},-\tfrac{1}{2}\hb^{-1}\al^2)\\
\bigg|\frac{\partial\mathcal{K}}{\partial\z}(\z,t)\bigg| & \leq
\sqrt{2\hb^{-1}}
\frac{\esf(t\sqrt{2\hb^{-1}},-\tfrac{1}{2}\hb^{-1}\al^2)}
{\esf(\z\sqrt{2\hb^{-1}},-\tfrac{1}{2}\hb^{-1}\al^2)}
\nsf(\z\sqrt{2\hb^{-1}},-\tfrac{1}{2}\hb^{-1}\al^2)
\msf(t\sqrt{2\hb^{-1}},-\tfrac{1}{2}\hb^{-1}\al^2)
\end{align*}
and similarly
\be\nn
\frac{\partial^{2}\mathcal{K}}{\partial\z^{2}}(\z,t)=
(2\hb^{-1})^{\frac{3}{2}}\z\mathsf{K}(\z,t).
\ee
All these estimates allow us to solve the equation (\ref{eq-append}) 
by applying Theorem \ref{int_eq1}. Using the notation of that theorem
we have
\begin{align*}
\phi(t) & =\frac{\ps(\z,\al,\hb)}{\Om(\z\sqrt{2\hb^{-1}})}\\
\psi_{1}(t) & =0\\
J(t) & =U(t\sqrt{2\hb^{-1}},-\tfrac{1}{2}\hb^{-1}\al^2)\\
\mathsf{K}(\z,t) & =-\frac{1}{2}
\frac{(\pi\hb)^{\frac{1}{2}}}{\G\big(\frac{1}{2}+\frac{1}{2}\hb^{-1}\al^2\big)}
\Omega(t\sqrt{2\hb^{-1}})\mathcal{K}(\z,t)\\
Q(t) & =
\frac{1}{2}
\frac{(\pi\hb)^{\frac{1}{2}}}{\G\big(\frac{1}{2}+\frac{1}{2}\hb^{-1}\al^2\big)}
\Omega(t\sqrt{2\hb^{-1}})
\esf(t\sqrt{2\hb^{-1}},-\tfrac{1}{2}\hb^{-1}\al^2)
\msf(t\sqrt{2\hb^{-1}},-\tfrac{1}{2}\hb^{-1}\al^2)\\
P_0(\z) & =
\frac
{\msf(\z\sqrt{2\hb^{-1}},-\tfrac{1}{2}\hb^{-1}\al^2)}
{\esf(\z\sqrt{2\hb^{-1}},-\tfrac{1}{2}\hb^{-1}\al^2)}\\
P_1(\z) & =\sqrt{2\hb^{-1}}
\frac
{\nsf(\z\sqrt{2\hb^{-1}},-\tfrac{1}{2}\hb^{-1}\al^2)}
{\esf(\z\sqrt{2\hb^{-1}},-\tfrac{1}{2}\hb^{-1}\al^2)}\\
\Phi(\z) & =\var_{0,\z}[H](\al,\hb)\\
\kappa_0 &\leq 
\frac{1}{2}(\pi\hb)^{\frac{1}{2}}l(-\tfrac{1}{2}\hb^{-1}\al^2)
\end{align*}
where the role of $\beta$ is played here by $Z$ and $\kappa$ is replaced 
for simplicity by the upper bound $\kappa_0$. Then the bounds 
(\ref{bound1}) and (\ref{bound2}) follow from Theorem \ref{int_eq1}. 

Finally, observe that all the integrals which
occur in the analysis above, converge uniformly when $\al\in[0,\delta]$
and $\z$ lies in any compact interval of $[0,Z)$; allowing us to state 
that $\epsilon_1$ and its first two partial $\z$-derivatives are 
continuous in $\al$ and $\z$. Consequently, the same stands 
for $\mathcal{Y}_1$ which signifies the end of the proof.
\end{proof}

\section*{Data Availability}

Data sharing is not applicable to this article as no new data were
created or analyzed in this study.

\section*{Acknowledgements} 

We are grateful to a referee for insisting  on the clarification of the 
results and proofs of section \S\ref{near-zero-evs}. The first author 
acknowledges the support of the Institute of Applied and Computational 
Mathematics of the Foundation of Research and Technology - Hellas 
(\href{https://www.forth.gr/}{FORTH}), via grant 
MIS 5002358. Also, the first author expresses his sincere gratitude to 
the Independent Power Transmission Operator 
(\href{https://www.admie.gr/en}{IPTO}) for a scholarship through the
\href{http://www.sse.uoc.gr/en/}{School of Sciences and Engineering} 
of the University of Crete.

\end{document}